\documentclass[11pt,leqno,textwidth=10cm]{amsart}
\usepackage{amssymb}
\allowdisplaybreaks
\usepackage{mathrsfs}
\usepackage{mathtools}
\usepackage{abstract}
\usepackage[hidelinks]{hyperref}
\usepackage{marginnote}
\usepackage{color}

\usepackage{etoolbox}
\makeatletter
\patchcmd{\@maketitle}{\newpage}{}{}{} 
\makeatother
\numberwithin{equation}{section}

\newtheorem{thm}{Theorem}[section]
\newtheorem{lem}[thm]{Lemma}
\newtheorem{prop}[thm]{Proposition}
\newtheorem{corol}[thm]{Corollary}

\theoremstyle{definition}
\newtheorem{defn}[thm]{Definition}

\theoremstyle{remark}

\newcommand{\p}{\partial}

\newcommand{\cR}{\mathcal{R}}
\newcommand{\cS}{\mathcal{S}}
\newcommand{\cW}{\mathcal{W}}

\newcommand{\define}{:=}
\newcommand{\temp}{\bar{t}^{\,}{}}
\newcommand{\tK}{\varkappa}

\newcommand{\alg}[1]{\begin{aligned}#1\end{aligned}}

\newcommand{\eq}[1]{\begin{equation} #1 \end{equation}}

\begin{document}


\title[]{Stability of fluids in spacetimes with decelerated expansion}
\author[D.~Fajman, M.~Ofner, T.~Oliynyk, Z.~Wyatt]{David Fajman, Maximilian Ofner, Todd Oliynyk, Zoe Wyatt}

\date{\today}


\address{\hspace{-1cm}
\begin{tabular}[h]{lll}
David Fajman \& Maximilian Ofner, &Todd Oliynyk &  Zoe Wyatt \\
Faculty of Physics, &School of Mathematics& Department of Pure Mathematics   \\ 
University of Vienna, &9 Rainforest Walk& and Mathematical Statistics, \\
Boltzmanngasse 5, &Monash University VIC 3800& Wilberforce Road,  \\
1090 Vienna, Austria. &Australia& Cambridge, CB3 0WB, U.K. \\
David.Fajman@univie.ac.at \& &todd.oliynyk@monash.edu& zoe.wyatt@maths.cam.ac.uk       \\
maximilian.ofner@univie.ac.at & &
\end{tabular}
}

\maketitle
\begin{abstract}
We prove the nonlinear stability of homogeneous barotropic perfect fluid solutions in fixed cosmological spacetimes undergoing decelerated expansion. The results hold provided a specific inequality between the speed of sound of the fluid and the expansion rate of spacetime is valid. Numerical studies in our earlier complementary paper provide strong evidence that the aforementioned condition is sharp, i.e.~that instabilities occur when the inequality is violated. In this regard, our present result covers the regime of slowest possible expansion which allows for fluids to stabilize, depending on their speed of sound. Our proof relies on an energy functional which is universal in the sense that it also applies to the case of linear expansion and enables a significantly simplified proof of bounds for fluids on linearly expanding spacetimes. Finally, we consider the special cases of dust and radiation fluids in the decelerated regime and prove shock formation for arbitrarily small perturbations of homogeneous solutions.
\end{abstract}

\section{Introduction}

\subsection{The relativistic Euler equations}
We consider the relativistic Euler equations
\eq{\label{EE}\alg{
\nabla^{\mu}T_{\mu\nu}&=0 \,,\quad 
T_{\mu\nu}=(\rho+p)u_{\mu}u_{\nu}+p g_{\mu\nu}\,,
}} 
 with a linear, barotropic equation of state $p=K\rho$ where $c_s=\sqrt K$ is the speed of sound of the fluid. Equations \eqref{EE} consistute one of the main matter models used in cosmology since the origin of its rigorous study over a century ago \cite{Friedman-1922}.  Depending on the value of $K$, the equations describes radiation ($K=1/3$), dust ($K=0$) or a fluid with speed of sound interpolating between both values. We consider the initial value problem for \eqref{EE} on the fixed cosmological spacetimes $ (M,g) $ 
\eq{\label{spacetime}
M=[t_{0},\infty)\times \mathbb T^3,  g=-dt\otimes dt+a(t)^2 \delta_{ij} dx^i\otimes dx^j,
}
where $(\mathbb T^3,\delta)$ is the standard 3-dimensional torus. Here $a:[t_{0},\infty)\rightarrow \mathbb R_+$ is the scale factor, which is a strictly increasing function of time. In the present paper, we restrict to a polynomial scale factor $a(t)=t^\alpha$ but we expect more general results could be derived in future work. We study the future global behaviour of fluids with sufficiently regular initial data given at $t_0>0$, i.e.~as $t\nearrow\infty$, which corresponds to the expanding direction of spacetime.

\subsection{Fluid stabilization}
The expansion of spacetime, $\dot a(t)>0$, generates a friction-type term in the equation for the fluid velocity (see \eqref{eq-v} below). In consequence, for sufficiently fast expansion and sufficiently small initial data, shock formation is suppressed during the evolution and the fluid has future global solutions. We refer to this phenomena as \emph{fluid stabilization}. This effect was first discovered in the work of Brauer, Rendall and Reula  \cite{BrauerRendallReula}. 
 It is known due to a series of rigorous studies  over the last 30 years, pioneered by the work of Rodnianski and Speck \cite{RodnianskiSpeck:2013}, that in the regime of accelerated expansion $\ddot a>0$, where the stabilizing effect is very strong, fluids stabilize for all values of $K\in [0,1/3]$\cite{Speck:2012,LubbeKroon:2013,HadzicSpeck:2015,Oliynyk:CMP_2016,Friedrich:2017,Mondal21}. This holds even in the superradiative regime \cite{OliynykBigK21} and other types of fluids \cite{LeFlochWei21,LiuWei:AHP_2021}. On the other hand, the expansion-driven damping effect may be counteracted by the speed of sound, which, if sufficiently large, can cause shock formation.  For $K>1/3$ numerical evidence suggests that instabilities occur \cite{BMO23}, while particular solutions are stable despite the high speeds of sound \cite{MO23,FMO24}.

Looking at slower expansion rates,  in the regime of linear expansion  where $\ddot a=0$, it is known that radiation fluids ($K=1/3$) develop shocks in finite time \cite{Speck:2013}. In contrast, subradiative fluids ($K<1/3$) remain stable in the linear regime \cite{FOW:CMP_2021}, even in the presence of backreaction under the fully nonlinear Einstein--Euler equations \cite{fajman2021slowly,fajman2023stability}. In conclusion, in the linear regime,  the future global regularity of the fluid  depends in a non-trivial manner on both the expansion rate and the speed of sound. The present paper aims to continue such an investigation into the more physically relevant decelerated regime.

\subsection{The decelerated regime}
The decelerated regime is characterized by $\ddot a(t)<0$ which corresponds in the polynomial class $a(t)=t^{\alpha}$ to $\alpha<1$. In the spatially compact case, as considered here, dispersion does not occur. In consequence, other mechanisms of decay are necessary to stabilize the fluid, such as dilution caused by the expansion. In the decelerated regime this dilution effect is weak. This manifests in the expansion-normalized relativistic Euler equations (see \eqref{eq-v-and-L} below) by the appearance of weakly-decaying error terms. For example,
\[
\p_tL= \frac{1}{t^\alpha}\partial_j v^j+ \cdots \,.
\]
 Such terms can only be bounded by terms of the form $t^{-\alpha} \sqrt{E}$, where $E$ is an appropriate Sobolev-type energy of the solution. Here a small polynomial decay of the energy is insufficient to make these terms integrable in time as $1-\alpha$ is a priori not small. An energy decay of at least $t^{-1+\alpha-\varepsilon}$ for some $\varepsilon>0$ is necessary to obtain integrability. This is in contrast to the simpler case of linear expansion, where $\alpha=1$ and hence any small polynomial decay for the energy would make these terms integrable. This slow decay is the first obstruction to prove the decay of the perturbations.

The second difficulty results from the speed of sound of the fluid. The origin of decay of perturbations
is the friction-type term in the equation for the expansion normalized fluid velocity:
\[
\p_t v^i=-\frac{\alpha(1-3K)}{t}v^i + \ldots \]
The coefficient of this term, which roughly corresponds to the rate of decay of perturbations, vanishes for radiation and is maximal for dust. This corresponds to the interpretation that fluids become more stable for smaller speeds of sound. To achieve a necessary decay rate the speed of sound needs to be sufficiently small. 

The balance between these two mechanisms is the key to understand the critical relation between the expansion rate and the speed of sound that allows for fluid stabilization.

\subsection{Main Result}
The main result of the present paper states that under the condition
\eq{\label{stab-cond}
K<1-\frac2{3\alpha} \,,
}
the quiet fluid solution
\eq{
(L_q,\vec v_q)=(L_0,\vec 0)\,,
}
where $L_0$ is a constant, is asymptotically stable as $t\nearrow \infty$. Here, $(L,v)$ respectively denote the expansion normalized energy density and fluid velocity (see \eqref{transformation} below). Our result is the first stability result in the decelerated regime in the spatially compact setting for fluids with non-vanishing speed of sound. Moreover, numerical studies in our complementary work  suggest that if condition \eqref{stab-cond} is violated, then the corresponding quiet fluid solutions are unstable \cite{fajman2024phase}. In this sense, our result presented here is presumably sharp and the curve $K(\alpha)=1-2/(3\alpha)$ marks the transition between the stable and the unstable regime. The precise formulation of the theorem is given in Theorem \ref{thm-stability} below. 

Our proof uses a formulation of the Euler equations in expansion-normalized variables $(L,v)$ and  $L^2$-based energy functionals. The key idea is the construction of a corrected $L^2$-energy, where the correction term generates a decay-inducing term for the spatial gradient of the expansion normalized energy density. This term complements the corresponding one for the expansion normalized fluid velocity and in total yields a decay-inducing term for the energy. However, at the same time, the time derivative of the corrected $L^2$-energy generates a positive term proportional to the velocity part, which reduces the corresponding negative term. The optimal choice for the coefficient of the correction term in the energy balances these two contributions to maximize the resulting energy decay. This process is precisely where the curve \eqref{stab-cond} arises as a sufficient criterion for stability.

In addition, our energy functional also requires key cancellations to occur so that we can obtain a closed energy estimate in regards to the regularity. These additional terms, however, do not interfere with the decay mechanism.


\subsection{Fluids on linearly expansion spacetimes}
Uniform bounds on the fluid variables for the range $K\in(0,1/3)$ were first obtained in the irrotational case in \cite{FOW:CMP_2021}. This was later generalized to the general case in 
\cite{fajman2023stability} in the presence of non-trivial curvature and gravitational backreaction. In both cases a Fuchsian approach in combination with a transformation of variables was used to obtain sufficient bounds on the fluid variables. 

In Section \ref{sec:linear} we consider a fluid on a linearly expanding background spacetime. We demonstrate that the energy functional, which we discover in the decelerated regime, fulfils a sufficient energy estimate, after a small adaption to the case of linear expansion. This provides a significantly shorter proof of the key steps in the result \cite{fajman2023stability}. This observation suggests that this energy is indeed canonical for relativistic fluids in a cosmological scenario.


\subsection{Some results on shock formation}
In Section \ref{sec:instability} we address the problem of shock formation for the special cases of dust $(K=0)$ and radiation $(K=1/3)$. In each case we show that there exists initial data arbitrarily close to the respective quiet fluid state, which leads to solutions that develop shocks in finite time. For radiation, this concerns all decelerated expansion rates $\alpha\leq 1$, which was first established in \cite{Speck:2013}. The proof given in the present article is constructive and makes use of characteristic coordinates. By contrast, the proof in \cite{Speck:2013} relies on the conformal invariance of the relativistic Euler equations when $K=1/3$ in order to apply the shock formation result of \cite{Christodoulou:2007}. For dust, our shock formation result concerns expansion rates $\alpha\leq1/2$, which complements the range $\alpha>1/2$ where stability is known \cite{Speck:2013}. Thus, we show that the phase transition for dust occurs at $\alpha=1/2$. Remarkably, in comparison with the location of the phase transition for fluids with $K>0$ at $\alpha=2(1-K)/3$, this suggests that there is no continuous limit in the fluid behaviour as $K\searrow 0$ for expansion rates $1/2<\alpha<2/3$. In that range dust is stable, while the results of \cite{fajman2024phase} suggest that for any $K>0$ the corresponding fluid is unstable.

\subsection{Structure of the paper}
In Section \ref{sec:prel} we present the relativistic Euler equations in their standard form and introduce the expansion-normalized rescaling and the relevant equations of motion. In Section \ref{sec:energy} we present the energy functional and a correction mechanism, which obeys the key energy estimate for solutions of the expansion normalized system. This estimate is used in Section 
\ref{sec:main} to prove the main theorem of this paper, Theorem \ref{thm-stability}. This section also contains a corollary stating the decay rates in the original physical variables. In Section \ref{sec:linear} we adapt the energy functional to the case of linearly expanding spacetimes and give a new short proof of fluid stabilization for the case $K \in (0,1/3)$. Finally, in Section \ref{sec:instability} we prove shock formation for dust and radiation fluids in the ranges $\alpha\leq 1/2$ and $\alpha<1$, respectively.


\section{Notation, coordinates and definitions}

We denote spacetime indices by Greek letters $ \mu, \nu, \dots $, while Roman letters $ i,j,\dots $ denote spacial indices. For spatial submanifolds $ \{t=\text{constant}\}\eqqcolon\Sigma_{t}\cong\mathbb{T}^{3} $ slices we choose standard Euclidean coordinates $ (x^{1},x^{2},x^{3}) $. For a vector field $ v $ on $ \Sigma_{t} $ with components $ v^{i} $ we  simply write
\begin{equation*}
	v_{i}\coloneqq \delta_{ij}v^{j} \text{\quad and \quad}\partial^{i}\coloneqq \delta^{ij}\partial_{j}.
\end{equation*}
In general, for a spacetime function $ f: M\to \mathbb{R} $ we  write
\begin{equation}
	\int f \coloneqq \int_{\mathbb{T}^{3}} f(\cdot,x)d^{3}x.
\end{equation}
To define higher-order norms, we adapt the following notation.

\begin{defn}[Spatial derivatives]
	Let $ f,g $ be functions, $ v,w $ be vector fields on $ \Sigma_{t} $ and $ \ell \in\mathbb{N} $. Then we write
	\begin{equation*}
		D^{\ell}_{k_{1}\cdots k_{\ell}}\coloneqq\partial_{k_{1}}\cdots \partial_{k_{\ell}}. 
	\end{equation*}
	In addition, we define the 
	product $ (\cdot,\cdot) $ via 
	\begin{equation*}
		(D^{\ell}f,D^{\ell}g)\coloneqq \partial_{k_{1}}\cdots\partial_{k_{\ell}}f\partial^{k_{1}}\cdots\partial^{k_{\ell}}g
	\end{equation*}
	as well as
	\begin{equation*}
		(D^{\ell}v,D^{\ell}w)\coloneqq \partial_{k_{1}}\cdots\partial_{k_{\ell}}v^{i}\partial^{k_{1}}\cdots\partial^{k_{\ell}}w_{i}.
	\end{equation*}
	Furthermore, we allow for constructions like
	\begin{equation*}
		(D^{\ell}v,D^{\ell+1}f)\coloneqq \partial_{k_{1}}\cdots\partial_{k_{\ell}}v^{i}\partial^{k_{1}}\cdots\partial^{k_{\ell}}\partial_{i}f.
	\end{equation*}
	In addition we denote $ |D^{\ell}\cdot|^{2}=(D^{\ell}\,\cdot\,,D^{\ell}\,\cdot) $, and $ D^{1} $ will simply be shortened to $ D $. 
\end{defn}

\begin{defn}[Norms]
	Let $ f $ be a function or tensor field on $ M $. We denote the inhomogeneous Sobolev-norm of order $ r $ by
	\begin{equation*}
		\|f\|^{2}_{H^{r}}\coloneqq \sum_{n=0}^{r}\int |D^{r}f|^{2}.
	\end{equation*}
	In addition, the Lebesgue-norm $ L^{2} $ is defined as 
	\begin{equation*}
		\|f\|_{L^{2}}^{2}\coloneqq \int |f|^{2}. 
	\end{equation*}
	
\end{defn}
In addition to standard Hölder- and Sobolev-type estimates we will frequently use the following.

\begin{lem}[Poincar\'e inequality]\label{Poincare}
	Let $ (\mathcal{M},\gamma) $ be a smooth, compact Riemannian manifold. Then
	\begin{equation*}
		\Big(\int_{\mathcal{M}} |f-\bar{f}|^{2}d\mu(\gamma)\Big)^{\frac{1}{2}}\lesssim \Big(\int |Df|^{2}d\mu(\gamma)\Big)^{\frac{1}{2}},
	\end{equation*}
	for all $ f\in H^{1}(\mathcal{M}) $, where $ \bar{f}=\frac{1}{\text{vol}(\mathcal{M},\gamma)}\int f d\mu(\gamma) $. 
\end{lem}

For a proof of Lemma \ref{Poincare}, see e.g.~\cite[Theorem 2.10]{Hebeynonlinear}. 

\section{Preliminaries}\label{sec:prel}

We consider the relativistic Euler equations \eqref{EE} on fixed cosmological FLRW spacetimes of the form \eqref{spacetime}. The equations take the explicit form
\begin{subequations}\label{Euler2}
\begin{align}
u^\mu \nabla_\mu \rho +(\rho+P)\nabla_\mu u^\mu &= 0,\label{Euler2_a}\\
(\rho+P)u^\mu\nabla_\mu u^\nu + (g^{\mu\nu}+u^\mu u^\nu ) \partial_\mu P &= 0, \quad \nu\in\{0,1,2,3\},\label{Euler2_b}
\end{align}
\end{subequations}
where the metric $g$ is given by \eqref{spacetime} and $\nabla$ denotes its Levi-Civita covariant derivative. The equation of state is used to replace the pressure via $
P=K\rho$. 
We then introduce expansion-normalized variables $(L,v)$  by
\eq{\label{transformation}
v^i =\frac{t^\alpha u^i}{\sqrt{1+t^{2\alpha}u^2}}\qquad\mbox{and}\qquad L=\log \left( t^{3\alpha(1+K)}\rho \right).
}
For the time-component of the fluid 4-velocity we use the normalization condition $u^\mu u_\mu=-1$ to derive
\eq{
u^0=\sqrt{1+t^{2\alpha}|u|^2}=\frac{1}{\sqrt{1-|v|^2}}\,.
}
In the expansion-normalized variables the system \eqref{EE} takes the following form.

\begin{subequations}\label{eq-v-and-L}
\eq{\label{eq-v}\alg{
\p_t v^i&=-\frac{\alpha(1-3K)}{t}v^i - \frac{K}{1+K}\frac{t^{1-\alpha}}{t}(1-v^2)\p_i L-\frac{t^{1-\alpha}}tv^j\p_j v^i
 +\frac{t^{1-\alpha}}{t}\left(1-\frac{1-K}{1-Kv^2}\right)v^i\p_j v^j
 \\
&\quad+\frac{t^{1-\alpha}}{t}\frac{1-K}{1+K}\left(1-\frac{1-K}{1-Kv^2}\right)v^i v^j\p_j L
+\frac{\alpha(1-3K)}{t}\frac{1-K}{1-Kv^2}v^2 v^i,
}}
\eq{\label{eq-L}\alg{
\p_tL&=-\frac{t^{1-\alpha}}{t}\frac{(1+K)}{1-Kv^2}\partial_j v^j-\frac{t^{1-\alpha}}{t}\frac{(1-K)}{1-Kv^2} v^j\p_j L+\frac{\alpha(1+K)}{t}\frac{1-3K}{1-Kv^2}v^2.
}}
\end{subequations}

The following identity is used later on:
\eq{\label{id-K}
1-\frac{1-K}{1-Kv^2}=K\frac{1-v^2}{1-Kv^2}.
}

We make the following choices and bootstrap assumptions for convenience:
\eq{\alg{\notag
&t>1\\
&|v|<1/10
}}

In the following sections we consider a solution $(L,v)$ emanating from initial data $(L_0,v_0)\in H^{N+1}(\mathbb T^3)\times H^{N+1}(\mathbb T^3)$ with $\|L_0-L_q\|_{H^{N+1}}+\|v_0\|_{H^{N+1}}<\varepsilon$ for $N \geq 6$ and $\varepsilon>0$ sufficiently small. Standard local-existence theory implies existence and uniqueness of a local-in-time solution to \eqref{eq-v-and-L}
\eq{\alg{
L:(t_0,t_1)\times \mathbb T^3&\rightarrow \mathbb R\\
v:(t_0,t_1)\times \mathbb T^3&\rightarrow \mathbb R^3,
}}
which can be extended beyond any $t_1>t_0$ if its Sobolev norm of order $N+1$ remains bounded as $t\nearrow t_1$ \cite{Speck:2013}.


\section{Energy estimates and the correction mechanism}
\label{sec:energy}

In this section we derive energy estimates which are the key to establish decay estimates for perturbations of $(L_0,\vec 0)$ for the lowest order of regularity. Subsequent higher order energy estimates follow the same structure. The mechanism to establish decay is apparent already at lowest order.

We begin with an estimate on the mean value of the velocity. Using the Poincar\'e inequality 
\eq{\notag
\|f-\bar f\|_{L^2}\leq C\|D f\|_{L^2}
}
where $f\in H^1(\mathbb T^3)$, we are able to control the $L^2$-norm of the velocity. We then define and analyze the $H^1$-energy of the variables $L$ and $\rho$, which obeys an energy estimate that is closed in regularity  as no higher order derivatives appear on the right-hand side. Ensuring that the estimate is closed in regularity is delicate. To upgrade this energy to a decay capturing energy functional we add a suitable correction term and establish the energy estimate for the corrected $H^1$-energy. 
We then commute the equations \eqref{eq-v}, \eqref{eq-L} with spatial derivatives to observe that the resulting equations have a similar structure to the uncommuted equations. This implies that, beside error terms, higher order energies fulfil identical estimates. Finally, we show that the higher order corrected energies are in fact coercive and hence control the solutions in the required Sobolev regularity.

\subsection{Evolution of the mean velocity}

We define the mean velocity field $\bar v^i$ by
\eq{
\overline v^i=\int v^i \,.
}
The mean velocity field obeys the following evolution equation.

\begin{lem}\label{lem-meanvelocity}
Let $i\in\{1,2,3\}$ and $v:(t_0,t_1)\rightarrow \mathbb R^3$ be a solution to \eqref{eq-v}. Then
\eq{\label{eq-mv}
\p_t\overline v^i=-\frac{\alpha(1-3K)}{t}\overline v^i+g^i(t),
}
where $g^i:(t_0,t_1)\times \mathbb R^3\rightarrow \mathbb R^3$ for $i\in\{1,2,3\}$ and
\eq{
|g(t,.)|\leq C\frac{1+t^{1-\alpha}}{t} p (\overline v^j,\|D v\|_{H^1},\|D v\|_{L^{\infty}},\|D L\|_{L^2})\,,
}
where $p(\cdot)$ is a polynomial with terms of at least second order (i.e.~vanishing zeroth and first order).
\end{lem}

\begin{proof}
The time derivative of the mean velocity field is derived using \eqref{eq-v}. We find
\eq{\alg{\notag
\p_t\overline v^i&=-\frac{\alpha(1-3K)}{t}\overline v^i-\frac{K}{1+K}t^{-\alpha}\left(\int \p_i L  -\int v^2\p_i L \right)\\
&\qquad-t^{-\alpha}\int v^j\p_jv^i  +t^{-\alpha}\int (1-\frac{1-K}{1-Kv^2})v^i\p_jv^j  \\
&\qquad+t^{-\alpha}\frac{1-K}{1+K}\int (1-\frac{1-K}{1-Kv^2})v^iv^j\p_j L  +\frac{\alpha(1-3K)}{t}\int \frac{1-K}{1-Kv^2}v^2 v^i  \,.
}}
The first term in the bracket in the first line vanishes due to integration by parts. For terms involving $v^i$-factors we decompose by $v^i=(v^i-\overline v^i)+\overline v^i$. For instance, for the second term in the bracket, we have
\eq{\alg{\notag
\int v^2\p_i L  &=\int (v_j-\overline v_j)(v^j-\overline v^j)\p_i L  +2\overline v^j\int (v_j-\overline v_j)\p_i L  +\overline v^j \overline v_j\int\p_i L   \,.
}}
The last term of this decomposition vanishes using integration by parts. The other terms in the decomposition can be estimated using the Poincar\'e inequality. We apply this argument to all terms except to the first one, which is listed explicitly in \eqref{eq-mv}.
\end{proof}

\subsection{$L^2$-energy}
In this section we define the canonical $L^2$-energy for solutions $(L,v)$ of \eqref{eq-v-and-L}. The specific structure of the energy is chosen so that it obeys an energy estimate which is closed in regularity.

\begin{defn}[$L^2$-energy]
The $L^2$-energy of a solution $(v,L)$ is defined as 
\eq{\label{def-en}\alg{
E_0[v,L]&=\frac12\int \partial_i v^j\partial^iv_j  +\frac{K}{1+K}\int v_m\p_iv^m\p ^iL \\
& \qquad+\frac12\int\frac{1}{1-v^2} v_m\p_i v^m v_n\p^i v^n  +\frac12\frac K{(1+K)^2}\int(1-v^2) \p^iL\p_iL \,.
}}

\end{defn}

In the derivation of the energy estimates for \eqref{def-en} the error terms are of the following type.

\begin{defn}[Perturbative term]
We call a term $f(t)$ \emph{perturbative}, if it can be bounded by
\eq{\label{f-est}
|f(t)| \leq C\frac{1+t^{1-\alpha}}{t} p(\|D L \|_{L^2},\|D v \|_{L^2},\|v\|_{L^{\infty}},\|D v\|_{L^{\infty}}),
}
where $p(\cdot)$ is a polynomial with lowest order terms of at least degree 3.
\end{defn}
We compute the time-derivatives of the individual terms of the energy in the following, keeping track of all important terms, while absorbing the remaining expressions into a perturbative term.

\begin{lem}\label{lem-en-p1}
Let $v$ be a solution to \eqref{eq-v}. Then, the following identity holds:
\eq{\alg{\notag
\p_t\Big(\frac12\int \partial_i v^j\partial^iv_j \Big) &=-\frac{\alpha(1-3K)}{t}\int \partial_i v^j\partial^iv_j  -\frac{K}{1+K}t^{-\alpha}\int \p_i v^j((1-v^2)\p^i\p_j L) \\
&+t^{-\alpha}\int \p^iv_j\left((1-\frac{1-K}{1-Kv^2})v^j \p_i\p_kv^k\right) \\
&+t^{-\alpha}\frac{1-K}{1+K}\int \p^i v_j\left((1-\frac{1-K}{1-Kv^2})v^j v^k\p_i\p_k L\right) +f(t)\,,
}}
where $f(t)$ is a perturbative term. 
\end{lem}

\begin{proof}
A straightforward computation using \eqref{eq-v} yields the following identity.
\eq{\alg{\notag
\p_t\frac12\int \partial_i v^j\partial^iv_j  &=-\frac{\alpha(1-3K)}{t}\int \partial_i v^j\partial^iv_j  -t^{-\alpha}\int \p^iv_j\p_i(v^k\p_k v^j) \\
&-\frac{K}{1+K}t^{-\alpha}\int \p_i v^j\p^i((1-v^2)\p_j L) \\
&+t^{-\alpha}\int \p^iv_j\p_i\left((1-\frac{1-K}{1-Kv^2})v^j \p_kv^k\right) \\
&+t^{-\alpha}\frac{1-K}{1+K}\int \p^i v_j\p_i\left((1-\frac{1-K}{1-Kv^2})v^j v^k\p_k L\right) \\
&+\frac{\alpha(1-3K)}{t}\int \p_i v^j\p_i\left(\frac{1-K}{1-Kv^2}v^2 v^j\right)  \,.
}}
The second term in the first line on the right-hand side is perturbative after commuting to obtain $\partial^iv_jv^k\partial_k\partial_i v^j$ and applying integration by parts. The term in the last line is immediately perturbative. In all remaining terms we evaluate the derivatives $\p_i$ and keep only the terms with second order derivatives, while absorbing the perturbative terms into $f(t)$. The resulting explicit terms give the identity in the lemma.
\end{proof}

\begin{lem}\label{lem-en-p2}
Let $(v,L)$ be a solution to \eqref{eq-v-and-L}. Then, the following identity holds:
\eq{\alg{\notag
\p_t\Big(\frac12\frac K{(1+K)^2}\int(1-v^2) \p^iL\p_iL \Big) &= -t^{-\alpha}\frac{K}{(1+K)}\int\frac{1-v^2}{1-Kv^2}\p^iL\p_i\p_jv^j +f(t)\,,
}}
where $f(t)$ is a perturbative term. 
\end{lem}

\begin{proof}
We use \eqref{eq-v} for the first factor yielding only perturbative terms, which are not explicitly listed below, but absorbed into $f(t)$. The other terms are given explicitly using \eqref{eq-L}.
\eq{\alg{\notag
&\p_t\Big(\frac12\frac K{(1+K)^2}\int(1-v^2) \p^iL\p_iL \Big) \\
&\qquad= \frac{K}{(1+K)^2}\Big[-t^{-\alpha}\int(1-v^2)\p^iL\p_i\left(\frac{1+K}{1-Kv^2}\p_jv^j\right) \\
&\qquad\quad-t^{-\alpha}\int(1-v^2)\p^iL\p_i\left(\frac{1-K}{1-Kv^2}v^j\p_jL\right) \\
&\qquad\quad+\frac{\alpha(1+K)}{t}\int(1-v^2)\p^iL\p_i\left(\frac{1-3K}{1-Kv^2}v^2\right) \Big]+f(t)
}}
The derivative $\p_i$ in the first term is evaluated and only the terms with second derivatives are kept explicitly providing the explicit term in the Lemma. The second term is perturbative by an integration by parts argument. The third term is perturbative immediately.
\end{proof}

\begin{lem}\label{lem-en-p3}

Let $(v,L)$ be a solution to \eqref{eq-v-and-L}. Then, the following identity holds:
\eq{\alg{\notag
\p_t\Big(&\frac{K}{1+K}\int v_m\p_iv^m\p ^iL\Big) \\
&=\frac{K}{1+K}\Big[-\frac{\alpha(1-3K)}{t}\int v^m\p^iL\p_i v_m +t^{-\alpha}\int v^m\p_iv_m\p^i(v^j\p_jL)\left(1-\frac{1-K}{1-Kv^2}\right) \\
&\qquad+t^{-\alpha}\int v^2\p^iL\p_i\p_jv^j\left(1-\frac{1-K}{1-Kv^2}\right) -t^{-\alpha}\int v_m\partial^iv^m\p_i\p_jv^j\left(\frac{1+K}{1-Kv^2}\right) \Big]+f(t)\,,
}}
where $f(t)$ is a perturbative term. 
\end{lem}

\begin{proof}
We use \eqref{eq-v} and \eqref{eq-L} to compute the time-derivatives of the second and first factor. By \eqref{eq-v} it is immediate that the terms resulting from the first factor are all perturbative.
\eq{\alg{\notag
\p_t&\Big(\frac{K}{1+K}\int v_m\p_iv^m\p ^iL\Big) \\
&=\frac{K}{1+K}\Big[-\frac{\alpha(1-3K)}{t}\int v^m\p^iL\p_i v_m  
-\underbrace{\frac{K}{1+K} t^{-\alpha}\int v^m\p^iL\p_i\left((1-v^2)\p_mL\right)}_{(i)} 
\\
&-\underbrace{t^{-\alpha}\int v^m\p^iL\p_i(v^j\p_jv_m)}_{(iii)} 
+\underbrace{t^{-\alpha}\int v^m\p^iL\p_i\left((1-\frac{1-K}{1-Kv^2})v_m\p_jv^j\right)}_{(iv)} \\
&
+\underbrace{t^{-\alpha}\int v^m\p^iL\frac{1-K}{1+K}\p_i\left((1-\frac{1-K}{1-Kv^2})v_mv^j\p_jL\right)}_{(i)} 
+\underbrace{\frac{\alpha(1-3K)}{t}\int v^m\p^iL\p_i\left(\frac{1-K}{1-Kv^2}v^2v_m\right)}_{(ii)} \\
&
-\underbrace{t^{-\alpha}\int v_m\partial^iv^m\p_i\left(\frac{1+K}{1-Kv^2}\p_jv^j\right)}_{(v)} -\underbrace{t^{-\alpha}\int v_m\partial^iv^m\p_i\left(\frac{1-K}{1-Kv^2}v^j\p_jL\right)}_{(iii)} 
\\
&
+\underbrace{\frac{\alpha(1+K)}{t}\int v_m\partial^iv^m\p_i\left(\frac{1-3K}{1-Kv^2}v^2\right)}_{(ii)} \Big]+f(t)\\
}}
The first term on the right-hand side appears as the first term on the right-hand side in the lemma. The terms marked $(i)$ are perturbative by an integration by parts argument. The terms marked $(ii)$ are directly perturbative. We combine the terms in $(iii)$ to obtain the second term on the right-hand side in the lemma up to perturbative terms. The term $(iv)$ yields the third term on the right-hand side of the lemma up to perturbative terms. Finally the term marked $(v)$ yields the last explicit term in the lemma, again up to perturbative terms. 
\end{proof}

\begin{lem}\label{lem-en-p4}

Let $(v,L)$ be a solution to \eqref{eq-v}-\eqref{eq-L}. Then the following identity holds. 
\eq{\alg{\notag
\p_t\Big(&\frac12\int\frac{1}{1-v^2} v_m\p_i v^m v_n\p^i v^n\Big)  =-\frac{\alpha(1-3K)}{t}\int\frac{v_m\p^iv^m}{1-v^2}v_n\p_iv^n \\
&-\frac{K}{1+K}t^{-\alpha}\int v_m\p^iv^m \p_i\left(v^n\p_nL\right)\frac{1-v^2}{1-Kv^2} +t^{-\alpha}\int\frac{v_m\p^iv^m}{1-v^2}\p_i(\p_jv^j)v^2\left(1-\frac{1-K}{1-Kv^2}\right) +f(t)
}}

\end{lem}

\begin{proof}
We use \eqref{eq-v} to evaluate the resulting terms when the time-derivative hits the derivative terms. The lower order terms yield only perturbative terms. This implies
\eq{\alg{\notag
\p_t\frac12&\Big(\int\frac{1}{1-v^2} v_m\p_i v^m v_n\p^i v^n \Big)
=-\frac{\alpha(1-3K)}{t}\int\frac{v_m\p^iv^m}{1-v^2}v_n\p_iv^n
\\&
-\underbrace{\frac{K}{1+K}t^{-\alpha}\int\frac{v_m\p^iv^m}{1-v^2}v_n\p_i\left((1-v^2)\p^nL\right)}_{(iv)}
+\underbrace{\frac{1-K}{1+K}t^{-\alpha}\int\frac{v_m\p^iv^m}{1-v^2}v_n\p_i\left((1-\frac{1-K}{1-Kv^2})v^nv^j\p_jL\right)}_{(iv)} 
\\
&
-\underbrace{t^{-\alpha}\int\frac{v_m\p^iv^m}{1-v^2}v_n\p_i(v^j\p_jv^n)}_{(i)}
+\underbrace{t^{-\alpha}\int\frac{v_m\p^iv^m}{1-v^2}v_n\p_i\left((1-\frac{1-K}{1-Kv^2})v^n\p_jv^j\right)}_{(iii)} \\
&
+\underbrace{\frac{\alpha(1-3K)}{t}\int\frac{v_m\p^iv^m}{1-v^2}v_n\p_i\left(\frac{1-K}{1-Kv^2}v^nv^2\right)}_{(ii)} +f(t).
}}

The first term on the right-hand side appears in the lemma. The term $(i)$ is perturbative by an integration by parts argument. The last explicit term $(ii)$ is perturbative. The term $(iii)$ equals the last explicit term in the lemma up to perturbative terms. The terms marked $(iv)$ equal the second term in the lemma up to perturbative terms, where we used the identity \eqref{id-K}.
\end{proof}

We have now evaluated the time-derivatives of all parts of the energy functional defined in \eqref{def-en}. We now combine these in the following proposition, which provides the identity for the $L^2$-energy.
\begin{prop}\label{en-est-L2}
Let $(v,L)$ be a solution to \eqref{eq-v}, \eqref{eq-L}. Then the following identity holds. 
\eq{\alg{\notag
\p_t E_0[v,L]&=-\frac{\alpha(1-3K)}{t}\int \partial_i v^j\partial^iv_j  +f(t)\,,
}}
where $f(t)$ is a perturbative term. 
\end{prop}

\begin{proof}
Combining Lemmas \ref{lem-en-p1} to \ref{lem-en-p4} yields the following identity:
\begin{align*}
\p_t E_0[v,L]&=-\frac{\alpha(1-3K)}{t}\int \partial_i v^j\partial^iv_j  \underbrace{-\frac{K}{1+K}t^{-\alpha}\int \p_i v^j((1-v^2)\p^i\p_j L) }_{(i)}\\
&+\underbrace{t^{-\alpha}\int \p^iv_j\left((1-\frac{1-K}{1-Kv^2})v^j \p_i\p_kv^k\right) }_{(ii)}
+\underbrace{t^{-\alpha}\frac{1-K}{1+K}\int \p^i v_j\left((1-\frac{1-K}{1-Kv^2})v^j v^k\p_i\p_k L\right) }_{(iii)}\\
&\underbrace{-t^{-\alpha}\frac{K}{(1+K)}\int\frac{1-v^2}{1-Kv^2}\p^iL\p_i\p_jv^j }_{(i)}\\
&+\frac{K}{1+K}\Big[-\frac{\alpha(1-3K)}{t}\int v^m\p^iL\p_i v_m  
+\underbrace{t^{-\alpha}\int v^m\p_iv_m\p^i(v^j\p_jL)\left(1-\frac{1-K}{1-Kv^2}\right) }_{(iii)}
\\
&
+\underbrace{t^{-\alpha}\int v^2\p^iL\p_i\p_jv^j\left(1-\frac{1-K}{1-Kv^2}\right) }_{(i)}
+\underbrace{-t^{-\alpha}\int v_m\partial^iv^m\p_i\p_jv^j\left(\frac{1+K}{1-Kv^2}\right) }_{(ii)}\Big]\\
&-\frac{\alpha(1-3K)}{t}\int\frac{v_m\p^iv^m}{1-v^2}v_n\p_iv^n 
+ \underbrace{-\frac{K}{1+K}t^{-\alpha}\int v_m\p^iv^m \p_i\left(v^n\p_nL\right)\frac{1-v^2}{1-Kv^2} }_{(iii)}
\\&
+\underbrace{t^{-\alpha}\int\frac{v_m\p^iv^m}{1-v^2}\p_i(\p_jv^j)v^2\left(1-\frac{1-K}{1-Kv^2}\right) }_{(ii)}+f(t)\,.
\end{align*}
The first term on the right-hand side appears in the lemma explicitly. The terms marked by $(i)$, after an integration by parts combine to zero (up to perturbative terms) using the following identity:
\eq{\alg{\notag
\frac{K}{1+K}\Big((1-v^2)&-\frac{1-v^2}{1-Kv^2}+v^2\Big(1-\frac{1-K}{1-Kv^2}\Big)\Big)
=0\,.
}}
The terms marked by $(ii)$ combine to zero (up to perturbative terms), again after an integration by parts and using  the following identity:
\eq{\alg{\notag
&1-\frac{1-K}{1-Kv^2}-\frac{K}{1+K}\frac{1+K}{1-Kv^2}+\frac{v^2}{1-v^2}\Big(1-\frac{1-K}{1-Kv^2}\Big)
=0\,.
}}
Finally, we collect the terms marked by $(iii)$, which combine to zero up to perturbative terms using the following identity:
\eq{\alg{\notag
&\frac{1-K}{1+K}\Big(1-\frac{1-K}{1-Kv^2}\Big)+\frac{K}{1+K}\Big(1-\frac{1-K}{1-Kv^2}\Big)-\frac{K}{1+K}\frac{1-v^2}{1-Kv^2}
=0\,.
}} All other unmarked terms are perturbative.
\end{proof}


\subsection{Corrected $L^2$-energy} Proposition \ref{en-est-L2} does not include a decay-inducing term on the right-hand side. Therefore, we next construct a correction term which we add to the $L^2$-energy in order to obtain a decay estimate. Despite the decay-inducing term we need to check that the resulting corrected energy remains coercive and hence equivalent to the standard first Sobolev norm of the fluid variables.
We define the correction term by:
\eq{\label{correction}
t^{\alpha-1}\int v^i\partial_i L  .
}
The following lemma provides its energy identity.
\begin{lem}
Let $(v,L)$ be a solution to \eqref{eq-v-and-L}. Then, the following identity holds:
\eq{\alg{\notag
\p_t\Big(t^{\alpha-1}\int v^i\partial_i L  \Big)&=-t^{\alpha-2}(1-3K\alpha)\int v^i\p_i L \\
&\qquad-\frac{1}{t}\frac{K}{1+K}\int \p^i L\p_iL  +\frac{(1+K)}{t}\int (\p_j v^j)^2 +f(t).
}}
\end{lem}

\begin{proof}
We use the equations \eqref{eq-v} and \eqref{eq-L} and observe that most terms that appear are perturbative due to the low-regularity of the correction term. We keep explicitly only the terms proportional to the correction terms and the quadratic terms in the energy.
\end{proof}

%

\begin{defn}[Corrected $L^2$-energy]
For $c\in\mathbb R$, we define
\eq{\label{corr-L2-en}
E_{0,c}[v,L]:=E_0[v,L]+c t^{\alpha-1}\int v^i\p_iL  .
}
\end{defn}

\begin{lem}\label{lem-cor-zero}
Let $(v,L)$ be a solution to \eqref{eq-v-and-L}. With $c=1/2(1+K)^{-1}\alpha(1-3K)$ the following identity holds:
\eq{\alg{\notag
\p_tE_{0,c}[v,L]&=-\frac{\alpha(1-3K)}{t}\Big[\frac{1}{2}\int \partial_i v^j\partial^iv_j  +\frac{1}{2}\frac{K}{(1+K)^2}\int \p^i L\p_iL  \Big]\\
&\qquad-\frac12 \frac{\alpha(1-3K)}{t}\int (\mathrm{rot}v)^2  -ct^{\alpha-2}(1-3K\alpha)\int v^i\p_i L +f(t)\,,
}}
where $f(t)$ denotes a perturbative term.
\end{lem}

For the proof we need the following standard lemma.

\begin{lem}
Let $v$ be a $C^1$-vector field and define $(\mathrm{rot} v)^2:=[\mathrm{rot} v]^j[\mathrm{rot} v]_j$. Then
\eq{
\int (\mathrm{rot} v)^2  = \int |D v|^2 -\int (\mathrm{div} v)^2 \,.
}
\end{lem}

\begin{proof}[Proof of Lemma \ref{lem-cor-zero}]
From the previous lemmas we obtain the following identity for arbitrary $c\in\mathbb R$.
\eq{\alg{\notag
\p_tE_{0,c}[v,L]&=-\Big[\frac{2\alpha(1-3K)-2c(1+K)}{t}\Big]\frac{1}{2}\int \partial_i v^j\partial^iv_j  
-2c(1+K)\frac1t\frac{1}{2}\frac{K}{(1+K)^2}\int \p^i L\p_iL  \\
&\qquad-\frac12 \frac{\alpha(1-3K)}{t}\int (\mathrm{rot}v)^2  +ct^{\alpha-2}(\alpha-1)\int v^i\p_i L +f(t)\,.
}}
We obtain the desired result with the choice
\eq{\notag
c=\frac12\frac{1}{1+K}\alpha(1-3K)\,.
}
\end{proof}


\subsection{Higher order energies}
In this section, we now consider the evolution equations after commutation with a suitable number of spatial derivatives.
We derive higher order energy estimates based on the structure for the $L^2$-energy \eqref{corr-L2-en}. This set up is based on the commuted equations \eqref{eq-v-higher order} and \eqref{eq-L-higher order}, which takes a similar form as the original equations with lower order terms, which can be treated uniformly.

\subsubsection{Commuted equations}
\begin{lem}\label{comm-eq}
Let $t>1$ and $(v,L)$ be a solution to \eqref{eq-v-and-L}. Furthermore, let $I\in \mathbb N^{3}$ denote a multi-index $I=(\ell_1,\ell_2,\ell_{3})$ and $D^I=\p^{\ell_1}_x\p^{\ell_2}_y\p^{\ell_3}_z$. 
Then, the following identity holds:
\eq{\label{eq-v-higher order}\alg{
\p_t (D^{I} v^i)&=-\frac{\alpha(1-3K)}{t}D^{I} v^i - \frac{K}{1+K}\frac{t^{1-\alpha}}{t}(1-v^2)\p^i D^{I} L\\
&\quad -\frac{t^{1-\alpha}}tv^j\p_j D^{I} v^i+\frac{t^{1-\alpha}}{t}\left(1-\frac{1-K}{1-Kv^2}\right)v^i\p_j D^{I} v^j\\
&\quad+\frac{t^{1-\alpha}}{t}\frac{1-K}{1+K}\left(1-\frac{1-K}{1-Kv^2}\right)v^i v^j\p_j D^{I} L\\
&\quad+\frac{\alpha(1-3K)}{t}\frac{1-K}{1-Kv^2}v^2 D^{I} v^i
+t^{-\alpha}R_v(v,D^{\leq I}v,D^{\leq I}L)\,,
}}
\eq{\label{eq-L-higher order}\alg{
\p_t(D^{I} L)&=-\frac{t^{1-\alpha}}{t}\frac{(1+K)}{1-Kv^2}\partial_j D^{I} v^j-\frac{t^{1-\alpha}}{t}\frac{(1-K)}{1-Kv^2} v^j\p_j D^{I} L\\
&\quad\,\,+\frac{\alpha(1+K)}{t}\frac{1-3K}{1-Kv^2} D^{I} v^2
+t^{-\alpha}R_L(v,D^{\leq I}v,D^{\leq I}L)\,,
}}
where
$R_v(.,.,.)$ and $R_L(.,.,.)$ are polynomials with at least second order terms, where the coefficients may contain the factor $1-\frac{1-K}{1-Kv^2}$. The maximum total number of derivatives appearing in each term is bounded by $|I|+1$.
\end{lem}

\begin{proof}
The lemma follows from commuting equations \eqref{eq-v-and-L} with $D^{I}$ and applying the Leibniz rule.
\end{proof}
\subsubsection{Higher-order energies}
We now define the higher-order energies with suitable correction term and derive their energy estimate. 

\begin{defn}[Higher-order energy]
Let $\ell\in\mathbb N$, we define
\eq{\alg{
E_\ell[v,L]&:=\frac12\int |D^{\ell+1}v|^2 +\frac{K}{1+K}\int v_{m}(D^{\ell+1}v^{m},D^{\ell+1}L) \\
& \qquad+\frac12\int\frac{1}{1-v^2} v_m v_{n}(D^{\ell+1}v^{m},D^{\ell+1}v^{n}) +\frac12\frac K{(1+K)^2}\int(1-v^2) |D^{\ell+1}L|^{2} \,.
}}
\end{defn}

\subsubsection{Higher-order energy identities}

We define the corrected higher-order energy analogously to what we did at lowest $L^2$-order. 
\begin{defn}[Higher-order corrected energy]
The higher-order correction term is given by
\eq{
C_\ell[v,L]=t^{\alpha-1}\int (D^{\ell}v,D^{\ell+1}L)\,.
} 
The higher-order corrected energy is defined by
\eq{\label{corr-energy}
E_{\ell,c}[v,L]:= E_\ell[v,L]+\frac12\frac{1}{1+K}\alpha(1-3K)C_\ell[v,L] \,.
}
\end{defn}

\begin{defn}[Higher-order perturbative term]
We call $R(t)$ a \emph{higher-order perturbative term}, if it can be bounded by
\eq{\alg{
|R(t)|\leq P(\|Dv\|_{H^N},\|DL\|_{H^N},\overline v),
}}
where $P(\cdot)$ denotes a polynomial where each term is at least of order three and $N\geq \max (6,\ell) $.
\end{defn}

\begin{lem}\label{lem-corr-ell}

Let $(v,L)$ be a solution to \eqref{eq-v-and-L}. For $c=\frac{\alpha(1-3K)}{2(1+K)}$ the following identity holds:
\eq{\alg{\notag
\p_t E_{\ell,c}[v,L]&=-\frac{\alpha(1-3K)}{t}\Big[\frac12\int |D^{\ell+1}v|^{2} +\frac12\frac K{(1+K)^2}\int(1-v^2) |D^{\ell+1}L|^{2}  \Big]\\
&\qquad- c(1-3K\alpha)t^{\alpha-2}\int(D^{\ell}v,D^{\ell+1}L)  -\frac12\frac{\alpha(1-3K)}{t}\int |D^{\ell}\mathrm{rot} v|^{2}+R(t)\,,
}}
where $R(t)$ is a higher-order perturbative term.
\end{lem}

\begin{proof}
Using Lemma \ref{comm-eq} to replace the time derivatives of the derivatives of $v$ and $L$, the identity follows in direct analogy to the proof of Lemma \ref{lem-cor-zero} and the lemmas leading up to it, since all non-perturbative terms have the same structure as in the aforementioned case. The perturbative terms arise from the terms $R_v$ and $R_L$ of Lemma \ref{comm-eq}. The perturbative terms involve expressions which contain at least two factors in either $v$ or derivatives of $v$ and $L$. The critical type of term that one needs to estimate involves precisely two such factors. In that case, we give as an example the method to estimate a term of the following form:
\eq{\label{example-higher-est}
\int |D^{\ell+1}v||D^{I}v||D^JL| ,
} 
where $I$, $J$ are multiindeces. Note that any coefficients, for example of the form $1-\tfrac{1-K}{1-Kv^2}$, can be estimated by $K$ from above.
Due to the condition that $|I|+|J|\leq \ell+2$,
 at least one of the terms (say $|J|$) fulfils the bound
\eq{\notag
|J|\leq\frac12(\ell+2)\,.
}  
This implies 
\eq{\notag
\|D^JL\|_{L^\infty}\leq C \| D L\|_{H^N}\,,
}
by Sobolov-embedding since $|J|+2\leq\frac12(\ell+2)+2\leq \frac 12 N+3\leq N$. 
Consequently, the integral \eqref{example-higher-est} can be estimated by the above and Cauchy-Schwarz to find
\eq{\notag
C\| D L\|_{H^N} \|D v\|^2_{H^N}\,.
} 
For the remaining higher order terms the derivatives are distributed over more than two terms and all but two terms can be estimated pointwise. 
Factors in $v$ are expanded via $v=(v-\overline v) +\overline v$. For example:
\eq{\notag
\int v|D^{\ell+1}v||D^{I}v||D^JL|=\int (v-\overline v)|D^{\ell+1}v||D^{I}v||D^JL|+\overline v\int |D^{\ell+1}v||D^{I}v||D^JL|.
} 
The second term is contained by the above discussion in the set of higher-order perturbative terms. In the first term $v-\overline v  $ is estimated pointwise and the resulting factor is estimated using the Poincar\'e inequality.
\end{proof}


\subsection{Equivalence of energies and total energy}

Before addressing the stability theorem it remains to show that the corrected energies in fact control the perturbation from the homogeneous background solution in standard Sobolev regularity. In order to obtain a uniform energy estimate we also define a total energy, which controls the solution at all orders as well as a corrected version thereof.

\begin{defn}[Total corrected energy]
We define the total $L^2$-based energy of order $N$ by
\eq{\alg{
E_N(t)&= \sum_{0\leq \ell\leq N} \frac12\int |D^{\ell+1}v|^{2} +\frac12\frac K{(1+K)^2}\int(1-v^2) |D^{\ell+1}L|^{2}\,.
}}
We define the total energy by
\eq{
\mathbf{E}_{N}(t):= \overline v(t)^2+ E_N(t)\,,
}
and we define the total corrected energy by
\eq{\label{total-energy}
\mathbf{E}_{N,c}(t):= \overline v(t)^2+ \sum_{0\leq\ell\leq N}E_{\ell,c}[v,L]\,.
}
\end{defn}

The following lemma shows that the corrected energy controls the standard energy.
\begin{lem}
For $t$ sufficiently large and $\|v\|_{L^{\infty}}$ sufficiently small there exists a positive constant $C$ such that
\eq{\alg{\notag
\frac1C\mathbf E_{N,c}(t)&\leq \mathbf E_{N}(t)\leq C \mathbf E_{N,c}(t) \,,\\
|\mathbf E_{N,c}(t)-\mathbf E_{N}(t)|&\lesssim t^{\alpha-1} \mathbf E_{N,c}(t)+\mathbf E_{N,c}(t)^{3/2}\,.
}}
\end{lem}

\begin{proof}
The lemma follows from applying first Cauchy-Schwarz on the correction terms to obtain 
\eq{\notag
|C_\ell[v,L]|\leq t^{\alpha-1}\|v\|_{H^\ell}\|L\|_{H^{\ell+1}}\leq \frac12 t^{\alpha-1}\Big(\|v\|_{H^\ell}^2+\|L\|^2_{H^{\ell+1}}\Big).
}
Then $t$ is chosen sufficiently large to assure that the right-hand side is smaller than the corresponding terms in the standard energy.  
\end{proof}



\section{Main result}
\label{sec:main}

The estimates derived above imply the following theorem, which is the main result of the present paper.

\begin{thm}\label{thm-stability}
Let $N\geq 6$, $0<K<1/3$ and $(1-K)\alpha>2/3$ and let $\mu>0$ be chosen such that  $\alpha(1-3K)-\mu>2(1-\alpha)$. Let $(v_0,L_0)\in H^{N+1}(\mathbb T^3)\times H^{N+1}(\mathbb T^3)$ be a vector field and a function, respectively. Then, there exists an $\varepsilon>0$ such that if the initial data $(v_0,L_0)$ satisfies 
\eq{
\overline v_0+ \|D v_0\|_{H^N}+\|D L_0\|_{H^N}<\varepsilon \,,
} 
then the solution $(v(t),L(t))$ to the system \eqref{eq-v-and-L}  exists to the future globally in time and the following decay rates hold:
\eq{\alg{
|\overline v(t)|&\leq C\varepsilon t^{(-\alpha(1-3K)+\mu)/2} \,,\\
\|L(t)-\overline L_0\|_{L^{\infty}}&\leq C\varepsilon \,,\\
\|D v(t)\|_{H^N}+\|D L(t)\|_{H^N}&\leq C\varepsilon t^{(-\alpha(1-3K)+\mu)/2} \,.
}}

\end{thm}


\begin{proof}

We consider now initial data with the conditions of Theorem \ref{thm-stability}, where $\varepsilon$ is the size of the initial data. 
We choose $t_0>1$ and $\varepsilon$ sufficiently small to assure that 
\eq{\alg{\notag
|v(t)|&<1
}}
on an interval $[t_0,T)$. Using Lemma \ref{lem-meanvelocity}, Lemma \ref{lem-cor-zero} and Lemma \ref{lem-corr-ell}
we obtain an inequality of the form
\eq{\alg{\notag
\p_t \mathbf{E}_{N,c}(t)&\leq-\frac{2\alpha(1-3K)}{t}\overline v(t)^2 -\frac{\alpha(1-3K)}{t} E_{N}(t)+C \frac{1+t^{1-\alpha}}{t} \mathbf E_N(t)^{3/2}+C \frac{t^{\alpha-1}}{t} \mathbf{E}_N(t),\\
&\leq -\frac{\alpha(1-3K)}{t} \mathbf E_{N,c}(t)+ C\frac{1+t^{1-\alpha}}{t} \mathbf E_{N,c}(t)^{3/2}+C \frac{t^{\alpha-1}}{t} \mathbf{E}_{N,c}(t)
}}
where we have absorbed any fourth-order or higher order terms into the third order term by assuming sufficient smallness and using the equivalence of energies.

Then for $\mu$ with $\alpha(1-3K)-\mu>2(1-\alpha)$, define 
\eq{
\widetilde{ \mathbf E}_{N,c}(t):= t^{\alpha(1-3K)-\mu}\mathbf E_{N,c}(t)\,.
} 
Hence,
\eq{\alg{\notag
\p_t\widetilde{ \mathbf E}_{N,c}(t)&\leq -\frac{\mu}{t}\widetilde{ \mathbf E}_{N,c}(t)+ C t^{\alpha(1-3K)-\mu}\left(\frac{1+t^{1-\alpha}}{t} \mathbf E_{N,c}(t)^{3/2}+ \frac{t^{\alpha-1}}{t} \mathbf{E}_{N,c}(t)
\right)\\
&= -\frac{\mu}{t}\widetilde{ \mathbf E}_{N,c}(t)+ \frac{C}t \left(t^{\frac{-\alpha(1-3K)+\mu}2}{(1+t^{1-\alpha}}) \widetilde{\mathbf E}_{N,c}(t)^{3/2}+ {t^{\alpha-1}} \widetilde{\mathbf{E}}_{N,c}(t)
\right).
}}
For sufficiently small initial data this implies 
\eq{\notag
\widetilde{ \mathbf E}_{N,c}(t)\leq \widetilde{ \mathbf E}_{N,c}(t_0)\,,
}
and in turn
\eq{\notag
{ \mathbf E}_{N,c}(t)\leq (t/t_0)^{-\alpha(1-3K)+\mu} { \mathbf E}_{N,c}(t_0)\,.
}
The theorem follows by the equivalence of energies. The bound for the rescaled energy density $L$ follows by using the decay estimates above in combination with an evolution equation for the mean value of $L$.
\end{proof}


We next translate the decay rates back into the original ``physical" variables $(\rho,u)$.
Inverting the transformation \eqref{transformation} yields
\eq{\alg{
u^j&=t^{-\alpha}\frac{v^j}{\sqrt{1-v^2}}, \quad
\rho=\frac{\exp (L)}{t^{3\alpha(1+K).}}
}}
We define the rescaled energy density by $\boldsymbol{\rho}(t)=\exp(L(t))$ and the rescaled velocity by $\boldsymbol{u}(t)=t^{\alpha}u(t)$.
\begin{corol}
The decay rates for the physical variables in Theorem \ref{thm-stability} are the following:
\eq{\alg{
|\overline {\boldsymbol{u}}(t)|&\leq C\varepsilon t^{(-\alpha(1-3K)+\mu)/2}\,,\\
|\boldsymbol{\rho}(t)-\boldsymbol{\rho}(t_0)|&\leq C\varepsilon\,,\\
\|\nabla \boldsymbol u(t)\|_{H^N}+\|\nabla \boldsymbol\rho(t)\|_{H^N}&\leq C\varepsilon t^{(-\alpha(1-3K)+\mu)/2}\,.
}}
\end{corol}


\section{Implications for fluids on linearly expanding cosmologies}
\label{sec:linear}
We show in this section that an adapted version of the corrected $L^2$-energy \eqref{corr-energy} can be applied to analyze the behaviour of fluids with $K\in (0,1/3)$ on linearly expanding cosmological spacetimes. In consequence, we obtain an alternative proof of the main stability theorems in \cite{FOW:CMP_2021} and \cite{fajman2023stability}.
In comparison with the decelerated regime, where $\alpha<1$, in the case of linear expansion the power law rate is $\alpha=1$. This implies that the error terms, which are of higher order in the energy estimates, now have a coefficient of $t^{-1}$. To compensate these terms in the corresponding energy estimates, it is sufficient to show \emph{any} non-trivial decay of the energy. This is crucial, since the correction term in the energy might otherwise violate coercivity on account of the fact that the decaying time-factor in the correction term  (see \eqref{correction}) is not present for $\alpha=1$.
We give a new proof of the following theorem, which was obtained in \cite{fajman2023stability} using Fuchsian methods.

\begin{thm}[Fluid stability under linear expansion]\label{thm-stability-linear}
Let $0<K<1/3$ and $\alpha=1$. Let $(v_0,L_0)\in H^{N+1}(\mathbb T^3)\times H^{N+1}(\mathbb T^3)$ be a vector field and a function, respectively. Then there exists an $\varepsilon>0$ and a $\delta>0$ such that for initial data $(v_0,L_0)$ satisfying
\eq{
\overline v_0+ \| D v_0\|_{H^N}+\| D L_0\|_{H^N}<\varepsilon,
} 
the solution $(v(t),L(t))$ to the system \eqref{eq-v-and-L}  exists to the future globally in time and the following decay rates hold:
\eq{\alg{
|\overline v(t)|&\leq C\varepsilon t^{-\delta/2} \,,\\
\|L(t)-\overline L_0\|_{L^{\infty}}&\leq C\varepsilon \,,\\
\| D v(t)\|_{H^N}+\| D L(t)\|_{H^N}&\leq C\varepsilon t^{-\delta/2}\,.
}}

\end{thm}

\begin{proof}
We sketch the proof in terms of the $L^2$-energy of lowest order. The higher order case works accordingly. We define the first corrected energy, where $\delta>0$ is to be determined later, by
\eq{
E^{\alpha=1}_{0,\delta}[v,L]:=E_{0}[v,L]+\delta \int _{\mathbb T^3} v^i \partial_i L ,
}
where $ E_{0} $ is defined in the same way as in \eqref{def-en}.

We choose $\delta>0$ sufficiently small to assure that the eventual total energy corresponding to \eqref{total-energy} is coercive, despite the fact that no decaying time-factor is present here. 
We are left to show a suitable decay estimate. We combine Proposition \ref{en-est-L2}, 
to obtain the following energy identity.
\eq{\alg{\notag
\p_t E^{\alpha=1}_{0,\delta}[v,L]&=-\frac{1-3K}{t}\int \partial_i v^j\partial^iv_j  -\delta t^{-1}(1-3K)\int v^i\p_i L \\
&\qquad-\frac{\delta}{t}\frac{K}{1+K}\int \p^i L\p_iL  +\delta\frac{(1+K)}{t}\int (\p_j v^j)^2 +\tilde f(t),
}}
where $\tilde f(t)$ obeys an estimate of the form \eqref{f-est} with $\alpha=1$. 
We deduce from the above identity 
\eq{
\alg{\notag
\p_t E^{\alpha=1}_{0,\delta}[v,L]&\leq-\frac{\delta/2}{t}E_{0,\delta}^{\alpha=1}[v,L]
-\Big(\frac{1-3K}{t}-\frac{\delta(1+K)}t-\frac{\delta/4}t\Big)\int \partial_i v^j\partial^iv_j  \\
&\qquad+\Big(\frac{\delta^2/2}{t}-\delta t^{-1}(1-3K)\Big)\int v^i\p_i L
-\frac{\delta}{t}\frac{K}{1+K}\Big(1-\frac{1}{4}\frac{1}{(1+K)}\Big)\int \p^i L\p_iL  +\tilde f(t).
}
}
Here, we have ignored the additional decay inducing term for the rotation of $v$. In addition, higher order terms like the second and third in \eqref{def-en} have been absorbed into $ \tilde{f}(t) $. 
For sufficiently small $\delta$ the indefinite term can be absorbed by the manifestly negative terms by using Cauchy-Schwarz, Poincar\'e's inequality and (possibly) Young's inequality. In combination with the structure of the term $\tilde f(t)$ this implies an estimate of the form
\eq{\notag
\p_t E^{\alpha=1}_{0,\delta}[v,L]\leq-\frac{\delta/2}{t}E_{0,\delta}^{\alpha=1}[v,L]+\frac Ct E^{\alpha=1}_{0,\delta}[v,L]^{3/2}. 
}
This implies the desired decay rate.
\end{proof}



\section{Instability for dust and radiation in the decelerated regime}
\label{sec:instability}
In this section we show that quiet dust and radiation solutions develop instabilities, in the form of shocks, on fixed FLRW spacetimes \eqref{spacetime} when $\alpha\leq 1/2$ and $\alpha\leq 1$, respectively.

\subsection{Instability for radiation fluids}
In this section we show shock formation for radiation fluids for specific arbitrarily small perturbations of homogeneous data. We use the formulation of the Euler equations as balance laws, as presented by Rendall and St\aa hl \cite{RendallStahl2008}.  Note that our data is small, in contrast to \cite{RendallStahl2008} where initial data is \textit{a priori} large.\\

%

We consider again spacetimes of the form \eqref{spacetime}, where we rename the $t$-variable by $\bar t$ and choose $a(\bar t)=\bar t^\alpha$ for $\alpha\in(0,1]$. Hence, in the coordinates $ (\bar{t},x^{i}) $, we have that
\begin{equation}\label{gbar}
	\bar{g}=-d\bar{t}^{\,2}+a(\bar{t})\delta_{ij}dx^{i}dx^{j}.
\end{equation}
In these coordinates we compute the second fundamental form $ \bar{k} $ in the standard coordinate frame $ dx^{i} $ to be 
\begin{equation}\notag
	\bar{k}_{ij}=-\frac{1}{2N}\partial_{\temp}\bar{g}_{ij}=-\frac12 \partial_{\temp} (\temp^{2\alpha}\delta_{ij}) = -\alpha \temp^{2\alpha-1}\delta_{ij}\,,
\end{equation}
where $ N=1 $ is the lapse function.
The mean curvature is $\text{tr}(\bar{k})=\temp^{-2\alpha}\delta^{ij} \bar{k}_{ij}=-3\alpha/\temp$. 
Thus we introduce the new constant mean curvature time $ t $ via
\begin{equation}\notag
	t:=-\frac{3\alpha}{\temp}, \quad t \in \big[-\frac{1}{3\alpha},0\big)
\end{equation}
so that $t\nearrow  0$ is the expanding direction. 
This gives $dt = \frac{3\alpha}{\temp^2} d\temp$ and so the metric \eqref{gbar}, expressed in coordinates $ (t,x^{i}) $, takes the form 
\begin{equation}\notag
	g=-\frac{9\alpha^2}{t^4}dt^2 +\frac{(3\alpha)^{2\alpha}}{|t|^{2\alpha}}\delta_{ij} dx^i dx^j.
\end{equation}
Hence, again in the coordinate frame, 
\begin{equation}\notag
	k_{ij}=-\frac{1}{2N}\partial_{t}g_{ij}=-|t|^{-2\alpha+1}\frac{(3\alpha)^{2\alpha}}{3}\delta_{ij}\,,
\end{equation} 
and we can check $\text{tr}(k) = -|t|=t$. 
Matching with \cite[eq. (1)]{RendallStahl2008}, where the metric takes the form
\eq{\notag
g=-N^2dt^2+A^2[(dx+\beta dt)^2+a^2(dy^2+dz^2)]
}
we find by comparison
\begin{equation}\notag
	N=\frac{3\alpha}{t^2},\quad A = \frac{(3\alpha)^\alpha}{|t|^\alpha},\quad \beta = 0,\quad a = 1.
\end{equation}
Additionally, the orthonormal frame introduced in \cite[eq. (2)]{RendallStahl2008}, becomes in our setting:
\begin{equation}\notag
	e_0 = \frac{t^2}{3\alpha}\p_t, \quad e_1 = \frac{|t|^{\alpha}}{(3\alpha)^\alpha}\p_x, \quad e_2 = \frac{t^\alpha}{(3\alpha)^\alpha}\p_y, \quad e_3 = \frac{t^\alpha}{(3\alpha)^\alpha}\p_z\,.
\end{equation}
We compute the second fundamental form $ \tilde{k} $ in the co-frame $ (e^{i}) $ to be:
\begin{equation}\notag
	k_{ij}=\tilde{k}_{kl}e^{k}{}_{i}e^{l}{}_{j}=\tilde{k}_{ij}\frac{(3\alpha)^{2\alpha}}{|t|^{2\alpha}} \Rightarrow \tilde{k}_{ij}=\frac{t}{3}\delta_{ij}.
\end{equation}
Following \cite[eq. (3)]{RendallStahl2008} we write the second fundamental form in $ (e^{i}) $ as
\begin{equation}\label{rad:Kdef}
	\tilde{k}_{ij}=-\frac12(\tK-t)\delta_{ij}+\frac12(3\tK-t)\delta_{i1}\delta_{1j}
\end{equation}
Comparing with \eqref{rad:Kdef}, we read off that $ \tK=\frac{t}{3} $. 


We assume a linear equation of state $p=K \rho$, and define the \emph{fluid enthalpy} $ \varphi $ by
\begin{equation}\notag
	 \varphi = \int_{\rho_c}^\rho \frac{\sqrt{K}}{1+K} \frac{dm}{m} = \frac{\sqrt{K}}{1+K} \ln(\rho/\rho_c)\,,
\end{equation}
where $\rho_c>0$ is a constant. Also following \cite[eq (7)]{RendallStahl2008}, we define the velocity parameter $ u $ via
\begin{equation}\notag
	u^0=(1-u^2)^{-1/2},\quad u^1 = u(1-u^2)^{-1/2},\quad u^2=u^3=0 \,.
\end{equation}
Observing that $e_1(A)=0=e_1(N)$ and following \cite[eq (14)]{RendallStahl2008}, we define the derivative operators
\begin{align*}
    D_+ &= e_0 + \frac{u+\sqrt{K}}{1+u\sqrt{K}}e_1=\tfrac{t^2}{3\alpha}\p_t + \frac{u+\sqrt{K}}{1+u\sqrt{K}}\Big(\frac{|t|}{3\alpha}\Big)^\alpha \p_x\,, \\ 
    D_- &= e_0 + \frac{u-\sqrt{K}}{1-u\sqrt{K}}e_1=\tfrac{t^2}{3\alpha}\p_t + \frac{u-\sqrt{K}}{1-u\sqrt{K}}\Big(\frac{|t|}{3\alpha}\Big)^\alpha \p_x\,,
\end{align*}
and from \cite[eq (15)]{RendallStahl2008} 
\begin{align*}
    r &= \varphi + \frac12 \ln{\frac{1+u}{1-u}} = \ln\Bigg(\left(\frac{\rho}{\rho_c}\right)^{\frac{\sqrt{K}}{1+K}}\Big( \frac{1+u}{1-u}\Big)^{1/2} \Bigg)\,,\\
    s &= \varphi - \frac12 \ln{\frac{1+u}{1-u}} = \ln\Bigg(\left(\frac{\rho}{\rho_c}\right)^{\frac{\sqrt{K}}{1+K}}\Big( \frac{1-u}{1+u}\Big)^{1/2} \Bigg)\,.
\end{align*}
Using $ \tK=\frac{t}{3} $, the balance laws from \cite[eq. (16)]{RendallStahl2008} then reduce to
\begin{equation}\label{balance_law1}
	D_+r = t\frac{\sqrt{K}+u/3}{1+\sqrt{K} u}\,, \quad D_-s = t\frac{\sqrt{K}-u/3}{1-\sqrt{K} u}\,.
\end{equation}

\subsubsection{Transformation to $R, S$}
First, we note that 
\begin{equation}\label{rad:changeofvar} 
	\rho = \rho_c 	\exp\left( \frac{1+K}{2\sqrt{K}}(r+s)\right)\,, \quad u = \frac{e^r-e^s}{e^r+e^s}\,.
\end{equation}
Motivated by this, we introduce the variables $R = e^r$ and $S = e^s$. Multiplying \eqref{balance_law1} by $e^r$ and $e^s$ respectively (as well as $(3\alpha)/t^2$) we use \eqref{balance_law1} to arrive at the equations
\begin{equation}\label{balance_law2}\begin{split}
\p_t R + \kappa(R,S) \frac{(3\alpha)^{1-\alpha}}{|t|^{2-\alpha}} \p_x R &= \frac{3\alpha}{t} f_1(R,S)R\,, \\ 
\p_t S+ \lambda(R,S) \frac{(3\alpha)^{1-\alpha}}{|t|^{2-\alpha}} \p_x S&= \frac{3\alpha}{t} f_2(R,S)S\,,
\end{split}
\end{equation}
where
\begin{equation}\notag
f_1(R,S) \define \frac13 \frac{(1+3\sqrt{K})R + (3\sqrt{K}-1)S}{(1+\sqrt{K})R+(1-\sqrt{K})S}\,,
\quad
f_2(R,S) \define \frac13 \frac{(3\sqrt{K}-1)R + (3\sqrt{K}+1)S}{(1-\sqrt{K})R+(1+\sqrt{K})S}\,,
\end{equation}
and 
\begin{equation}\notag\begin{split}
\kappa(R,S) &\define\frac{u+\sqrt{K}}{1+\sqrt{K} u} = \frac{(1+\sqrt{K})R-(1-\sqrt{K})S}{(1+\sqrt{K})R+(1-\sqrt{K})S}\,, \\
\lambda(R,S) &\define \frac{u-\sqrt{K}}{1-\sqrt{K}u} =  \frac{(1-\sqrt{K})R-(1+\sqrt{K})S}{(1-\sqrt{K})R+(1+\sqrt{K})S} \,.
\end{split}\end{equation}
For later use, we compute the following 
\begin{equation}\label{derivs_kappa_lambda}
	\begin{split}
		\partial_{1} \kappa(x,y) &= \frac{2(1-K)y}{((1+\sqrt{K})x+(1-\sqrt{K})y)^2}\,, \quad 
		\partial_{2} \kappa(x,y) = \frac{-2(1-K)x}{((1+\sqrt{K})x+(1-\sqrt{K})y)^2}\,, \\
		\partial_{1} \lambda(x,y) &= \frac{2(1-K)y}{((1-\sqrt{K})x+(1+\sqrt{K})y)^2}\,, \quad
		\partial_{2} \lambda(x,y) = \frac{-2(1-K)y}{((1-\sqrt{K})x+(1+\sqrt{K})y)^2} \,.
\end{split}\end{equation}

\subsubsection{Coordinate change to physical time}

Now that we have the equations of motion  \eqref{balance_law2} for the fluid, we transform back to physical time $ \temp $. Denoting $ \bar{f}(\temp)=(f\circ t)(\temp) $, we find
\begin{equation}\label{rad:eomphysical}
	\begin{aligned}
	\partial_{\temp}\bar{R}+\kappa(\bar{R},\bar{S})\temp^{-\alpha}\partial_{x}\bar{R}&=-3\alpha \temp^{-1}f_{1}(\bar{R},\bar{S})\,,\quad 
	\partial_{\temp}\bar{S}+\lambda(\bar{R},\bar{S})\temp^{-\alpha}\partial_{x}\bar{S}&=-3\alpha \temp^{-1}f_{2}(\bar{R},\bar{S})\,.
	\end{aligned}
\end{equation}
In the case of radiation, a straightforward calculation shows that
\begin{equation}\notag 
	f_{1}(x,y)|_{K=\frac{1}{3}}=f_{2}(x,y)|_{K=\frac{1}{3}}=\frac{1}{\sqrt{3}}.
\end{equation}
Therefore, \eqref{rad:eomphysical} simplifies to 
\begin{equation}\label{rad:eomphysicalrad}
	\begin{aligned}
		\partial_{\temp}\bar{R}+\kappa(\bar{R},\bar{S})\temp^{-\alpha}\partial_{x}\bar{R}&=-\sqrt{3}\alpha \temp^{-1}\,, \qquad 
		\partial_{\temp}\bar{S}+\lambda(\bar{R},\bar{S})\temp^{-\alpha}\partial_{x}\bar{s}=-\sqrt{3}\alpha \temp^{-1}.
	\end{aligned}
\end{equation}

\subsubsection{Transformation to $\cR, \cS$}
We introduce
\begin{equation}\label{rad:caldef}
R \define t^{\sqrt{3}\alpha} \cR, \quad S \define t^{\sqrt{3}\alpha} \cS.
\end{equation}
Using \eqref{rad:eomphysicalrad} we find 
\begin{equation}\label{rad:eomrescaled}
	\begin{aligned}
		\partial_{\temp}\mathcal{R}+\kappa(\mathcal{R},\mathcal{S})\temp^{-\alpha}\partial_{x}\mathcal{R}&=0\,,\quad
		\partial_{\temp}\mathcal{S}+\lambda(\mathcal{R},\mathcal{S})\temp^{-\alpha}\partial_{x}\mathcal{S}=0 \,.
	\end{aligned}
\end{equation}
Note that we used that $ \kappa(c R,cS)=\kappa(R, S) $ as well as $ \lambda(cR,cS)=\lambda(x,y) $ for any $c\neq 0$. 

\subsubsection{Characteristic Coordinates}
We now introduce characteristic coordinates $ (\tau,\xi) $
\begin{equation}\label{rad:charcoor}
	(\temp,x) = (\tau, \phi(\tau, \xi))
\end{equation}
chosen so that
\begin{equation}\notag
 \phi(\tau_{0}, \xi) = \xi, \quad \p_{\tau} \phi(\tau, \xi) = \kappa(\tilde\cR(\tau,\xi),\tilde\cS(\tau,\xi)) \tau^{-\alpha}
\end{equation}
where for any function $f=f(t,x)$ we write
$$ \tilde{f}(\tau,\xi):= f(\tau, \phi(\tau,\xi))\,.
$$

\subsubsection{Deriving a Riccati-type equation for radiation}

Commuting \eqref{rad:eomrescaled} with $ \partial_{x} $ and denoting $ \partial_{x}\mathcal{R}=\mathcal{W} $ we find
\begin{equation}\notag
	\partial_{\temp}\mathcal{W}=-\temp^{-\alpha}\kappa(\mathcal{R},\mathcal{S})\partial_{x}\mathcal{W}-\temp^{-\alpha}(\partial_{1}\kappa)(\mathcal{R},\mathcal{S})\mathcal{W}^{2}-\temp^{-\alpha}(\partial_{2}\kappa)(\mathcal{R},\mathcal{S})\mathcal{W}\partial_{x}\mathcal{S} \,.
\end{equation}
Furthermore, we remember that due to  \eqref{rad:eomrescaled}, we have that
\begin{equation}\notag
	\partial_{\temp} \cS+ \temp^{-\alpha}\lambda(R,S) \p_x \cS= 0 \,.
\end{equation}
Note that this equation is solved by the constant solution $ \cS(t,x)=c $, where $ c\in \mathbb{R} $. Hence, if we assume that solutions $ (\cR,\cS) $ exist on an interval $ [\temp_{0},\temp_{1}) $, and $ \cS(\temp_{0})=c $, then $ \cS(\temp)=c $ for all $ \temp\in [
\temp_{0},\temp_{1}) $ and $ x $.  Thus, if $ \mathcal{S} $ is chosen to be homogeneous in $ x $ initially, then we have that $ \partial_{x}\mathcal{S} \equiv 0 $. Therefore, we are left with
\begin{equation}\notag
	\partial_{\temp}\mathcal{W}=-\temp^{-\alpha}\kappa(\mathcal{R},\mathcal{S})\partial_{x}\mathcal{W}-\temp^{-\alpha}(\partial_{1}\kappa)(\mathcal{R},\mathcal{S})\mathcal{W}^{2}\,.
\end{equation}
In coordinates $ (\tau,\xi) $ as introduced in \eqref{rad:charcoor}, we then find that
\begin{equation}\label{rad:riccati}
	\begin{aligned}
	\partial_{\tau}\tilde{\mathcal{W}}(\tau,\xi)&=\partial_{\tau}\mathcal{W}(\tau,\phi(\tau,\xi))=(\partial_{1}\mathcal{W})(\tau,\phi(\tau,\xi))+(\partial_{2}\mathcal{W})(\tau,\phi(\tau,\xi))\partial_{\tau}\phi(\tau,\xi)\\
	&=-\tau^{-\alpha}\kappa(\tilde{\mathcal{R}},\tilde{\mathcal{S}})(\tau,\xi)(\partial_{2}\mathcal{W})(\tau,\phi(\tau,\xi))-\tau^{-\alpha}(\partial_{1}\kappa)(\tilde{\mathcal{R}},\tilde{\mathcal{S}})\tilde{\mathcal{W}}^{2}(\tau,\xi)\\
	&\quad+\tau^{-\alpha}\kappa(\tilde{\mathcal{R}},\tilde{\mathcal{S}})(\tau,\xi)(\partial_{2}\mathcal{W})(\tau,\phi(\tau,\xi))\\
	&=-\tau^{-\alpha}(\partial_{1}\kappa)(\tilde{\mathcal{R}},\tilde{\mathcal{S}})\tilde{\mathcal{W}}^{2}(\tau,\xi)\,.
	\end{aligned}
\end{equation}

\subsubsection{Constructing initial data for blowup}
For some constant $ \mathring{\mathcal{R}}>0 $, consider initial data  
\begin{equation}\label{rad:initialquiet}
	\cR(1,\cdot)=\cS(1,\cdot)=\mathring{\mathcal{R}}>0 \,.
\end{equation} 
Invoking the translation formula \eqref{rad:changeofvar}, we find that this data coincides with
\begin{equation}\notag
	u(1,\cdot)=0\,, \quad \rho(1,\cdot) = \rho_{c}\mathring{\mathcal{R}}^{\frac{1+K}{\sqrt{K}}}\,.
\end{equation}
Hence, we see that for any $ \mathring{\mathcal{R}} $ we recover a quiet fluid solution, launched by this initial data. 

Now consider a non-trivial perturbation  $ \mathcal{R}(1,\cdot)=\mathring{\mathcal{R}}+h $, where $ h\in C^{\infty}(S^{1}) $ with $ h\neq 0 $. Then, there exists an $ x_{0}\in S^{1} $, such that $ \mathcal{W}(1,x_{0})<0 $. From \eqref{rad:riccati} as well as \eqref{derivs_kappa_lambda}, we have the following Riccati-type ODE for $ \tilde{\mathcal{W}}(\,\cdot\,,x_{0}) $:  
\begin{equation}\notag
	\frac{d}{d\tau}\tilde{\mathcal{W}}(\tau,x_{0})=\frac{12\tilde{\mathcal{S}}}{((3+\sqrt{3})\tilde{\mathcal{R}}-(-3+\sqrt{3})\tilde{\mathcal{S}})^{2}}\tau^{-\alpha}\cW^{2}(\tau,x_{0})\,.
\end{equation}
Note that $ \tilde{\mathcal{R}}(\cdot,x_{0}) $ is constant which is clear from \eqref{rad:eomrescaled} and the definition of $ (\tau,\xi) $ coordinates. In addition, $ \cS $ is constant everywhere for all time. This ODE exhibits finite-time blowup for all $ \alpha\leq 1 $ and, by Gr\"onwall's lemma, solution exists  for all time given sufficiently small initial data for $ \alpha>1 $.

	\subsection{Dust}
	Finally, we consider the case of dust. We again use the method of characteristics to establish shock formation for arbitrarily small perturbations of homogeneous solutions. We  consider spacetimes of the form \eqref{spacetime} subject to $\alpha\leq 1/2$. We remark that for $\alpha>1/2$ it has been shown in \cite{Speck:2013} that small perturbations of homogeneous dust solutions stabilize. Thus, our result shows that this result from \cite{Speck:2013} is sharp in the parameter $\alpha$. 
	\subsubsection{Setup}
	
	For the metric \eqref{spacetime} with $a(t)=t^{\alpha}$ the  equation  \eqref{Euler2_b} with $K=0$ reduces to
	\begin{equation}\notag
		u^{0}\partial_{0}u^{j}+u^{i}\partial_{i}u^{j}+2\alpha t^{-1}u^{0}u^{j}=0\,.
	\end{equation}
This is equivalent to 
	\begin{equation}\label{eomufull}
		\partial_{0}u^{j}+\frac{1}{u^{0}}u^{i}\partial_{i}u^{j}=-2\alpha t^{-1}u^{j}\,.
	\end{equation}
	We can easily see from \eqref{eomufull} that if we set $ u^{2}=u^{3}=0 $ initially, this condition propagates. Thus from now on we  write simply $ u^{1}=u $ and the equations simplify to 
	\begin{equation}\label{eomu}
		\partial_{0}u+\frac{1}{\sqrt{t^{2\alpha}u^{2}+1}}u\partial_{x}u=-2\alpha t^{-1}u.
	\end{equation}

	\subsubsection{Characteristics of the equations of motion}
	
	Assume that we have classical solution $ u(t,x) $ on some patch $ (1,T)\times \mathbb{T}^{3} $ and a curve $ x $ on said patch. The solution on this curve is given by
	\begin{equation*}
		z(s)=u(s,x(s))\,.
	\end{equation*}
	The system of \eqref{eomu} for $ t(1)=1 $, $ x(1)=x_{1} $ and $ z(1,x_{1})=g(x_{1}) $ is then equivalent to
	\begin{equation}\label{characteristics}
		\begin{aligned}
			\frac{dt}{ds}(s)&=1\Rightarrow t=s \,, \\
			\frac{dx}{dt}(t)&=\frac{z(t)}{\sqrt{t^{2\alpha}z(t)^{2}+1}}\,, \qquad \frac{dz}{dt}(t)=-2\alpha t^{-1}z(t)\,.
		\end{aligned}
	\end{equation}
	Integrating the last equation, we have
	\begin{equation*}
		z(t)=t^{-2\alpha}g(x_{1}).
	\end{equation*}
	Furthermore, we have that 
	\begin{equation}\label{xdot}
		\dot{x}(t)=\frac{g(x_{1})t^{-2\alpha}}{\sqrt{t^{-2\alpha}g(x_{1})^{2}+1}}\,.
	\end{equation}

	\subsubsection{Upper bounds for the characteristics}
	
	Now consider the solution to the initial value problem
	\begin{equation}\label{boundzeta}
		\begin{cases}
			\dot{\zeta}(t)=g(x_{1})t^{-2\alpha},\\
			\zeta(1)=x_{1}.
		\end{cases}
	\end{equation}
	Assuming that $ g(x_{1})>0 $, from the fact that
	\begin{equation}\label{boundcalc}
		|\dot{\zeta}(t)|=|g(x_{1})|t^{-2\alpha}>\frac{|g(x_{1})|t^{-2\alpha}}{\sqrt{t^{-2\alpha}g(x_{1})^{2}+1}}=|\dot{x}(t)|
	\end{equation}
	we can infer the bound $\zeta(t)>x(t)$
	for all $ t>1 $. 
	\subsection{Case $\alpha=1/2$}
	Consider the case $ \alpha =\frac{1}{2} $. Assume now without loss of generality that $ x_{1}<x_{2} $ and $ 0<g(x_{2})<g(x_{1}) $ (i.e.~an interval where the initial data is strictly decreasing). Then, the solutions  $ \zeta^{1},\zeta^{2} $ of \eqref{boundzeta}, with initial data $ x_{1} $ and $ x_{2} $ respectively, are given by
	\begin{equation}\label{solzeta}
		\begin{aligned}
		\zeta^{1}(t)&=g(x_{1})\log t+x_{1}\,,\\
		\zeta^{2}(t)&=g(x_{2})\log t+x_{2}. 
		\end{aligned}
	\end{equation}
	Equating the right-hand side, we find that
	\begin{equation}\notag
		\log t=-\frac{x_{2}-x_{1}}{g(x_{2})-g(x_{1})}>0.
	\end{equation}
	Hence, the two curves will meet at time
	\begin{equation}\label{breakingtime}
		T=\exp\left(-\frac{x_{2}-x_{1}}{g(x_{2})-g(x_{1})}\right).
	\end{equation}
	By the mean value theorem, we may conclude that there exists a breaking time $ T_{b} $, such that for any time $ T_{0}> T_{b} $ there exist $ x_{1,2} $, such that $ \zeta_{1,2} $ will have an intersection time $ T<T_{0} $. This breaking time is given by 
	\begin{equation}\label{Tcrit}
		T_{b}=\exp\left(-\frac{1}{\inf_{x}g^{\prime}(x)}\right).
	\end{equation} 
	Note, however, that the curves $ \zeta^{1,2} $ are only upper bounds for $ x^{1,2} $ (the solutions to \eqref{xdot} with the same initial data), but do not envelope them. Hence, nothing about the intersections of the characteristics of \eqref{eomu} can be said at this point. 
	
	\subsubsection{Construction of the enveloping curves}
	
	Integrating \eqref{xdot} in the case of $\alpha =\frac{1}{2} $ with initial data $ x_{1} $ yields
	\begin{equation}\notag
		x^{1}(t)=-2g(x_{1})\log(-\sqrt{t}+\sqrt{g(x_{1})^{2}+t})+c\,.
	\end{equation}
	Prescribing initial data yields
	\begin{equation}\notag
		c=x_{1}+2g(x_{1})\log(-1+\sqrt{g(x_{1})^{2}+1})\,,
	\end{equation}
	and hence
	\begin{equation}\notag
		x^{1}(t)=x_{1}+2g(x_{1})\log\Big(\frac{-1+\sqrt{g(x_{1})^{2}+1}}{-\sqrt{t}+\sqrt{g(x_{1})^{2}+t}}\Big)\,.
	\end{equation}
	Using \eqref{solzeta} we calculate the difference
	\begin{align*}
		|x^{1}(t)-\zeta^{1}(t)|&=\Big| 2g(x_{1})\log\Big(\frac{-1+\sqrt{g(x_{1})^{2}+1}}{-\sqrt{t}+\sqrt{g(x_{1})^{2}+t}}\Big)-2g(x_{1})\log(\sqrt{t})\Big|\notag \\
		&=\Big|2g(x_{1})\log\Big(\frac{-1+\sqrt{g(x_{1})^{2}+1}}{-t+\sqrt{tg(x_{1})^{2}+t^2}}\Big)\Big|\,.
	\end{align*}
	Now let us observe what happens to the denominator as $ t $ goes to infinity: 
	\begin{equation}\notag
		-t+\sqrt{tg(x_{1})^{2}+t^2}=\frac{tg(x_{1})^{2}+t^2-t^{2}}{t+\sqrt{tg(x_{1})^{2}+t^2}}=\frac{g(x_{1})^{2}}{1+\sqrt{t^{-1}g(x_{1})^{2}+1}}\overset{{t\to \infty}}{\longrightarrow}\frac{g(x_{1})^{2}}{2}.
	\end{equation}
	The difference $ |x^{1}(t)-\zeta^{1}(t)| $ is easily seen to be increasing and, by the calculation above, converges to
	\begin{equation}\label{difference}
		|x^{1}(t)-\zeta^{1}(t)|\overset{{t\to \infty}}{\longrightarrow}\Big| 2g(x_{1})\log\Big(\frac{2(-1+\sqrt{g(x_{1})^{2}+1})}{g(x_{1})^2}\Big)\Big|.
	\end{equation}
	Note that by l'Hopital, we have that
	\begin{equation}\label{gto0}
		\lim_{g(x_{1})\to 0}\frac{2(-1+\sqrt{g(x_{1})^{2}+1})}{g(x_{1})^2}=\lim_{g(x_{1})\to 0}\frac{2\frac{1}{2}\frac{1}{\sqrt{g(x_{1})^{2}+1}}2g(x_{1})}{2g(x_{1})}=1,
	\end{equation}
	and hence, for fixed $ t $, 
	\begin{equation}\notag
		|x^{1}(t)-\zeta^{1}(t)|\overset{{t\to \infty}}{\longrightarrow}0.
	\end{equation}
	Thus, we have found that, at least pointwise, by decreasing the initial data size we cause $ \zeta^{1} $ and $ x^{1} $ to converge. 
	
	Now since the the two curves $ \zeta $ and $ x $ cannot ever be further apart than the expression given in \eqref{difference}, we can construct a new curve
	\begin{equation}\notag
		\xi^{1}=\zeta^{1}-\Big|2g(x_{1})\log\Big(\frac{2(-1+\sqrt{g(x_{1})^{2}+1})}{g(x_{1})^2}\Big)\Big|=:\zeta^{1}-m(g(x_{1}))
	\end{equation}
	where $ m $ denotes the \emph{margin} between $ x^{1} $ and $ \zeta^{1} $. Note that, by construction, $ \xi^{1} $ is a lower bound for the characteristic $ x^{1} $.
	
	\subsubsection{Estimating the time of intersection}
	
	Since $ \xi^{1} $ is a lower bound of $ x^{1} $ and $ \zeta^{2} $ is an upper bound of $ x^{2} $, a crossing of these bounds at a point in time indicates that the characteristics $ x^{1,2} $ must have already intersected. 
	Equating $ \xi^{1}(t) $ and $ \zeta^{2}(t) $ yields
	\begin{equation}\notag
		g(x_{1})\log t+x_{1}-m(g(x_{1}))=g(x_{2})\log t+x_{2},
	\end{equation}
	which is equivalent to the statement that
	\begin{equation}\label{breakingxi}
		\log t=-\frac{x_{2}-x_{1}}{g(x_{2})-g(x_{1})}+\frac{m(g(x_{1}))}{g(x_{1})-g(x_{2})}.
	\end{equation}
	Therefore, the enveloping curves $ \zeta^{2} $ and $ \xi^{1} $ intersect precisely at 
	\begin{equation}\notag
		T=\exp\Big(-\frac{x_{2}-x_{1}}{g(x_{2})-g(x_{1})}\Big)\exp\Big(\frac{m(g(x_{1}))}{g(x_{1})-g(x_{2})}\Big).
	\end{equation}
	Taking the limit $ x_{1}\to x_{2} $ is not possible, as it increases the error significantly. However, for fixed $ x_{1,2} $, we see that by rescaling the initial data profile to be smaller ($ g\to \lambda g $), by \eqref{gto0} we have that
	\begin{equation}\notag
		\exp\Big(\frac{m(g(x_{1}))}{g(x_{1})-g(x_{2})}\Big)\overset{\lambda\to 0}{\longrightarrow} 1.
	\end{equation}
	Hence, small initial data ensures that shock happens later in time, but increases the accuracy of the prediction of the time of the forming of the shock.

	\subsection{Case $ \alpha< \frac{1}{2} $}
	
	We consider now  $ \alpha<\frac{1}{2} $. As before, we know that $ \zeta $, the solution to the initial value problem \eqref{boundzeta}, is an upper bound to the solution (given that $ g(x_{1})>0 $). \\
Consider the solution to the initial value problem
	\begin{equation}\label{boundgamma}
		\begin{cases}
			\dot{\gamma}(t)=\frac{g(x_{1})}{\sqrt{g(x_{1})^{2}+1}}t^{-2\alpha},\\
			\gamma(1)=x_{1}.
		\end{cases}
	\end{equation}
	By an analogous estimate to \eqref{boundcalc}, we see that 
	\begin{equation}
		\gamma(t)<x(t)
	\end{equation}
	for all $ t>1 $, given that $ g(x_{1})>0 $. 
	
	Consider again two solutions launched from $ x_{1,2} $ with $ 0<g(x_{2})<g(x_{1}) $. By the previous analysis, we see that $ \gamma_{1} $ is a lower bound to $ x^{1} $ and $ \zeta^{2} $ is a upper bound for $ x^{2} $. Equating these bounds we find the relation
	\begin{equation}\notag
		x_{2}+g(x_{2})\frac{t^{-2\alpha+1}-1}{-2\alpha-1}=	x_{1}+\frac{g(x_{1})}{\sqrt{g(x_{1})^{2}+1}}\frac{t^{-2\alpha+1}-1}{-2\alpha-1}.
	\end{equation}
	Solving for $ t $ yields a breaking time
	\begin{equation}\notag
		T=\Big((-2\alpha+1)(x_{2}-x_{1})\frac{1}{\frac{g(x_{1})}{\sqrt{g(x_{1})^{2}+1}}-g(x_{2})}+1\Big)^{\frac{1}{1-2\alpha}}.
	\end{equation}
	Note that all of the terms in this expression are strictly positive with the exception of 
	\begin{equation}
		\frac{g(x_{1})}{\sqrt{g(x_{1})^{2}+1}}-g(x_{2}).
	\end{equation}
	The expression above may become negative, if, for example, $ g $ does not vary much in $x$ but is large enough so that the denominator is dominant. 
		This can be remedied by the following observation: given any initial data profile $ g $, we define a new smaller data profile by $ \tilde{g}=\lambda g $, for some $ \lambda>0 $ to be determined later. Then, we calculate
	\begin{equation}\label{rescalecalc}
		\frac{\tilde{g}(x_{1})}{\sqrt{\tilde{g}(x_{1})^{2}+1}}-\tilde{g}(x_{2})=\frac{\lambda g(x_{1})}{\sqrt{\lambda^{2}g(x_{1})^{2}+1}}-\lambda g(x_{2})=\frac{\lambda (g(x_{1})-g(x_{2})\sqrt{\lambda^{2}g(x_{1})^{2}+1})}{\sqrt{\lambda^{2}g(x_{1})^{2}+1}}\,.
	\end{equation}
	Now, choose $ \lambda $ such that
	\begin{equation}\notag
		\lambda<\frac{1}{g(x_{1})}\sqrt{\frac{g(x_{1})^{2}}{g(x_{2})^{2}}-1} \,.
	\end{equation}
	Note that this is always possible as $ \frac{g(x_{1})}{g(x_{2})}>1 $. Then the bracket in the numerator in \eqref{rescalecalc}, and therefore the whole expression, is larger than zero. 
	We have therefore guaranteed that the expression for $ T $ exists and is larger than $ 1 $. Hence, we have shown that there exist arbitrarily small, smooth initial data that develop shock singularities in finite time. 

\subsection*{Acknowledgements}

D.F.~and M.O.~acknowledge support by the Austrian Science Fund (FWF) grant  \emph{Matter dominated Cosmology} 10.55776/PAT7614324 as well as \emph{Relativistic Fluids in Cosmology}  P 34313-N and thank the Erwin-Schr\"odinger Institute for hospitality during the program \emph{Nonlinear waves and Relativity}. Futhermore, M.O.~acknowledges support by the Austrian Science Fund (FWF) grant Y963 and recognizes that this work was in part funded within the Dimitrov Fellowship Program of the OeAW. 

	
\bibliographystyle{plain}
\bibliography{refs}

\begin{thebibliography}{10}

\bibitem{BMO23}
F.~Beyer, E.~Marshall, and T.~A. Oliynyk.
\newblock Future instability of flrw fluid solutions for linear equations of
  state $p=k\rho$ with $1/3<k<1$.
\newblock {\em Phys. Rev. D}, 107:104030, May 2023.

\bibitem{BrauerRendallReula}
U.~Brauer, A.~Rendall, and O.~Reula.
\newblock The cosmic no-hair theorem and the non-linear stability of
  homogeneous {N}ewtonian cosmological models.
\newblock {\em Classical Quantum Gravity}, 11(9):2283--2296, 1994.

\bibitem{Christodoulou:2007}
D.~Christodoulou.
\newblock {\em The Formation of Shocks in 3-Dimensional Fluids}.
\newblock EMS, 2007.

\bibitem{fajman2024phase}
D.~Fajman, M.~Maliborski, M.~Ofner, T.~Oliynyk, and Z.~Wyatt.
\newblock {Phase transition between shock formation and stability in
  cosmological fluids}.
\newblock 2024.
\newblock preprint [arXiv:2405.03431].

\bibitem{fajman2023stability}
D.~Fajman, M.~Ofner, T.~A. Oliynyk, and Zoe Wyatt.
\newblock The stability of relativistic fluids in linearly expanding
  cosmologies.
\newblock {\em Int. Math. Res. Not. IMRN}, (5):4328--4383, 2024.

\bibitem{fajman2021slowly}
D.~Fajman, M.~Ofner, and Z.~Wyatt.
\newblock Slowly expanding stable dust spacetimes.
\newblock {\em Arch.~Rat.~Mach.~Anal.}, 248(5):66, 2024.

\bibitem{FOW:CMP_2021}
D.~Fajman, T.A. Oliynyk, and Zoe Wyatt.
\newblock Stabilizing relativistic fluids on spacetimes with non-accelerated
  expansion.
\newblock {\em Commun. Math. Phys.}, 383:401--426, 2021.

\bibitem{FMO24}
G.~Fournodavlos, E.~Marshall, and T.~A. Oliynyk.
\newblock Future stability of perfect fluids with extreme tilt and linear
  equation of state $p=c_s^2\rho$ for the {E}instein-{E}uler system with
  positive cosmological constant: The range $\frac{1}{3}<c_s^2<\frac{3}{7}$,
  2024.

\bibitem{Friedman-1922}
A.~Friedman.
\newblock {\"U}ber die {K}r{\"u}mmung des {R}aumes.
\newblock {\em Zeitschrift f{\"u}r Physik}, 10(1):377--386, 1922.

\bibitem{Friedrich:2017}
H.~Friedrich.
\newblock Sharp asymptotics for {E}instein-$\lambda$-dust flows.
\newblock {\em Comm. Math. Phys.}, 350:803 -- 844, 2017.

\bibitem{HadzicSpeck:2015}
M.~Had\v{z}i\'{c} and J.~Speck.
\newblock The global future stability of the {FLRW} solutions to the
  {D}ust-{E}instein system with a positive cosmological constant.
\newblock {\em J. Hyper. Differential Equations}, 12:87--188, 2015.

\bibitem{Hebeynonlinear}
E.~Hebey.
\newblock {\em Nonlinear analysis on manifolds: {S}obolev spaces and
  inequalities}, volume~5 of {\em Courant Lecture Notes in Mathematics}.
\newblock New York University, Courant Institute of Mathematical Sciences, New
  York; American Mathematical Society, Providence, RI, 1999.

\bibitem{LeFlochWei21}
P.~G. LeFloch and C.~Wei.
\newblock Nonlinear stability of self-gravitating irrotational {C}haplygin
  fluids in a {FLRW} geometry.
\newblock {\em Ann. Inst. H. Poincar\'e{} C Anal. Non Lin\'eaire},
  38(3):787--814, 2021.

\bibitem{LiuWei:AHP_2021}
C.~Liu and C.~Wei.
\newblock Future stability of the {FLRW} spacetime for a large class of perfect
  fluids.
\newblock {\em Ann. Henri Poincar\'{e}}, 22(3):715--770, 2021.

\bibitem{LubbeKroon:2013}
C.~L\"{u}bbe and J.~A.~Valiente Kroon.
\newblock A conformal approach for the analysis of the non-linear stability of
  radiation cosmologies.
\newblock {\em Annals of Physics}, 328:1--25, 2013.

\bibitem{MO23}
E.~Marshall and T.~A. Oliynyk.
\newblock On the stability of relativistic perfect fluids with linear equations
  of state $p=k \rho$ where $1/3<k<1$.
\newblock {\em Letters in Mathematical Physics}, 113(5):102, 2023.

\bibitem{Mondal21}
P.~Mondal.
\newblock The linear stability of the {$n+1$} dimensional {FLRW} spacetimes.
\newblock {\em Classical Quantum Gravity}, 38(22):Paper No. 225009, 51, 2021.

\bibitem{Oliynyk:CMP_2016}
T.~A. Oliynyk.
\newblock Future stability of the {FLRW} fluid solutions in the presence of a
  positive cosmological constant.
\newblock {\em Commun. Math. Phys.}, 346:293--312; see the preprint
  [arXiv:1505.00857] for a corrected version, 2016.

\bibitem{OliynykBigK21}
Todd~A. Oliynyk.
\newblock Future global stability for relativistic perfect fluids with linear
  equations of state {$p=K\rho$} where {$1/3<K<1/2$}.
\newblock {\em SIAM J. Math. Anal.}, 53(4):4118--4141, 2021.

\bibitem{RendallStahl2008}
A.D. Rendall and F.~Stahl.
\newblock Shock waves in plane symmetric spacetimes.
\newblock {\em Communications in Partial Differential Equations},
  33:2020--2039, 2008.

\bibitem{RodnianskiSpeck:2013}
I.~Rodnianski and J.~Speck.
\newblock The stability of the irrotational {E}uler-{E}instein system with a
  positive cosmological constant.
\newblock {\em J. Eur. Math. Soc.}, 15:2369--2462, 2013.

\bibitem{Speck:2012}
J.~Speck.
\newblock The nonlinear future-stability of the {FLRW} family of solutions to
  the {E}uler-{E}instein system with a positive cosmological constant.
\newblock {\em Selecta Mathematica}, 18:633--715, 2012.

\bibitem{Speck:2013}
J.~Speck.
\newblock The stabilizing effect of spacetime expansion on relativistic fluids
  with sharp results for the radiation equation of state.
\newblock {\em Arch. Rat. Mech.}, 210:535--579, 2013.

\end{thebibliography}
\end{document}